\documentclass[conference]{IEEEtran}
\ifCLASSINFOpdf
\else
\fi
\usepackage{float}
\usepackage{graphicx, xcolor}
\usepackage{algorithm}
\usepackage{algpseudocode}
\usepackage{amsmath, amsthm}
\newtheorem{theorem}{Theorem}[section] 

\usepackage{subcaption}
\newtheorem{proposition}[theorem]{Proposition}
\newtheorem{lemma}[theorem]{Lemma}
\newtheorem{corollary}[theorem]{Corollary}
\newtheorem{definition}[theorem]{Definition}

\begin{document}
%
\title{The Hidden Cost of Correlation: Rethinking Privacy Leakage in Local Differential Privacy}




%




\author{\IEEEauthorblockN{Sandaru Jayawardana\IEEEauthorrefmark{1},
Sennur Ulukus\IEEEauthorrefmark{2},
Ming Ding\IEEEauthorrefmark{3}and
Kanchana Thilakarathna\IEEEauthorrefmark{1}}
\IEEEauthorblockA{\IEEEauthorrefmark{1}The University of Sydney, Australia}
\IEEEauthorblockA{\IEEEauthorrefmark{2}University of Maryland, USA}
\IEEEauthorblockA{\IEEEauthorrefmark{3} Data61, CSIRO, Australia}}



\maketitle

\begin{abstract}
Local differential privacy (LDP) has emerged as a promising paradigm for privacy-preserving data collection in distributed systems, where users contribute multi-dimensional records with potentially correlated attributes. Recent work has highlighted that correlation-induced privacy leakage (CPL) plays a critical role in shaping the privacy–utility trade-off under LDP, especially when correlations exist among attributes.
Nevertheless, it remains unclear to what extent the prevailing assumptions and proposed solutions are valid and how significant CPL is in real-world data.
To address this gap, 
we first perform a comprehensive \emph{statistical analysis} of five widely used LDP mechanisms---GRR, RAPPOR, OUE, OLH and Exponential mechanism---to assess CPL across four real-world datasets.
We identify that many primary assumptions and metrics in current approaches fall short of accurately characterising these leakages.
Moreover, current studies have been limited to a set of pure LDP (i.e., $\delta = 0$) mechanisms.
In response, 
we develop the first algorithmic framework to theoretically quantify 
CPL for any general approximated LDP ($(\varepsilon,\delta)$-LDP) mechanism. 
We validate our theoretical results against empirical statistical results and provide a theoretical explanation for the observed statistical patterns.
Finally, we propose two novel benchmarks to validate correlation analysis algorithms and evaluate the utility vs CPL of LDP mechanisms. 
Further, we demonstrate how these findings can be applied to achieve an efficient privacy–utility trade-off in real-world data governance.

\end{abstract}


%
\IEEEpeerreviewmaketitle

\section{Introduction}\label{section:introduction}

In modern distributed systems, multi-dimensional data---comprising multiple correlated attributes---is routinely collected from sources such as IoT devices~\cite{LDP-PCC-dependency-graph-IoT} and end users~\cite{apple2017privacy}. This data fuels applications including crowdsourced analytics~\cite{PCKV_key_value, apple2017privacy} and personalized services~\cite{rappor_original}. However, such datasets frequently include sensitive information, including personally identifiable information (PII)\cite{LDP_survey_composition_theorem} and individual preferences\cite{rappor_original}, raising significant privacy concerns. In scenarios where the data collector is untrusted or users demand strong privacy guarantees, Local differential privacy (LDP)~\cite{LDP_duchi} has emerged as a prominent solution. LDP enforces privacy at the user end, ensuring that each data record is randomised before transmission—thereby eliminating the need to trust the data aggregator. 
For example, LDP mechanisms (a mechanism is a randomised algorithm that perturbs the data) have been successfully deployed in the industry by major players such as Google~\cite{rappor_original}, Microsoft~\cite{microsoft_ldp}, and Apple~\cite{apple2017privacy} to collect user information in a private manner. 
LDP mechanisms are characterized by two parameters: the privacy budget $\varepsilon \geq 0$, which bounds the maximum privacy leakage, and the relaxation parameter $\delta \in [0, 1)$, which allows for rare privacy violations to improve utility~\cite{Algorithmic_Foundations_dwork}. Lower $\varepsilon$ implies stronger privacy (via more noise) but reduced utility. A small $\delta$ is preferred to maintain strong privacy. Thus, experts aim to find the best trade-off between data privacy and utility by carefully choosing the $\varepsilon$ and $\delta$. 

\begin{figure}[t]
  \centering
  \includegraphics[trim={1cm .6cm 0.5cm 0.8cm},clip,width=1.02\linewidth]{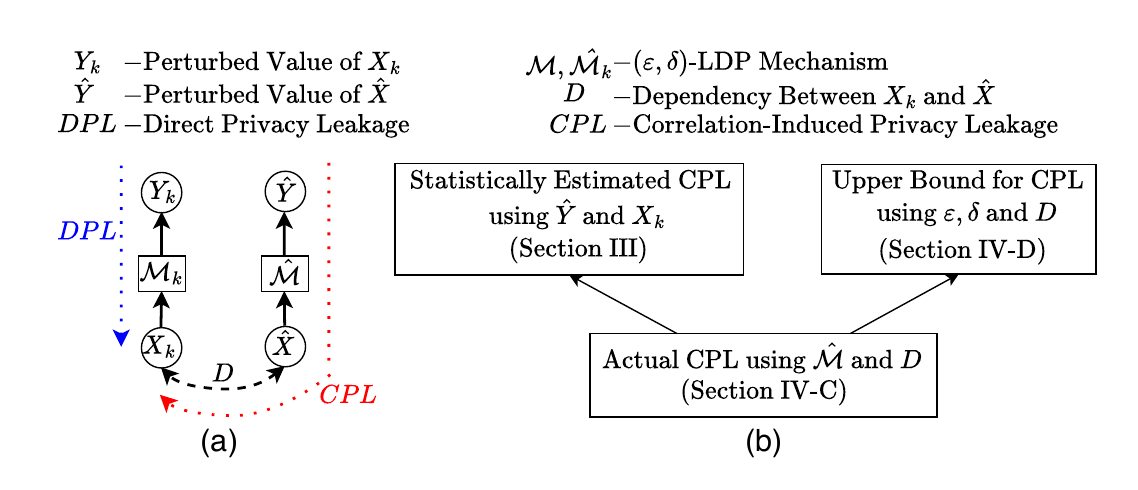}
  \vspace{-1.5em}
  \caption{(a) Example for privacy leakage analysis of $X_k$ attribute. (b) Proposed methods of computing CPL. 
  }
  \vspace{-1.8em}
  \label{fig:intro}
\end{figure}

However, collecting privacy-preserved multi-dimensional data poses significant challenges because most attributes/features of real-world data exhibit interdependence and increase privacy leakage~\cite{no_free_lunch}. 
For example, let $X_k$ and $\hat{X}$ be two correlated attributes, and let $Y_k$ and $\hat{Y}$ be perturbed values of $X_k$ and $\hat{X}$ generated by LDP mechanisms $\mathcal{M}_k$ and $\mathcal{\hat{M}}$, respectively as shown in Fig.~\ref{fig:intro}(a). Suppose we need to analyse the privacy leakage of $X_k$. Then, an adversary can learn about $X_k$ directly from $Y_k$. Additionally, if the adversary has the dependency information (e.g., from prior knowledge) between $X_k$ and $\hat{X}$ (i.e., $D$), then more information about $X_k$ can be inferred from $\hat{Y}$ through $\hat{X}$. Let us call \textbf{direct privacy leakage (DPL)} the theoretical maximum possible learning directly from $Y_k$, and  \textbf{correlation-induced privacy leakage (CPL)} the theoretical maximum possible learning through $\hat{Y}$ (See Section~\ref{subsection:additional_privacy_leakage} for definitions). Then, studies have proposed various approaches to apply LDP in multi-dimensional data, such as privacy budget splitting~\cite{SPL_RS_RS_PLUS_FAKE}, random-sampling~\cite{RS_plus_FD}, joint perturbation~\cite{PCKV_key_value}, etc, to overcome CPL impact. These solutions primarily assumed that the prior knowledge about the dependency between attributes is unknown.

However, when prior knowledge is available, leveraging it to analyse CPL enables a more nuanced trade-off between privacy and utility.
The literature shows that prior knowledge can be gained in various methods: from universal norms~\cite{impact_of_prior_knowledge}, from public datasets~\cite{adult_dataset, celebA_attributes}, learning while collecting data~\cite{AAA} etc. 
Few recent studies have used this prior knowledge to analyse the privacy leakage for LDP mechanisms in a correlated environment~\cite{data_correlation_location_similar_LDP,secon_Collecting_High-Dimensional_Correlation_LDP}. 
Study~\cite{secon_Collecting_High-Dimensional_Correlation_LDP} demonstrates that prior knowledge can be leveraged to calibrate the privacy budgets of $\varepsilon$-Univariate Dominance LDP (a variant of LDP) mechanisms to maximise the utility. 
Study~\cite{data_correlation_location_similar_LDP} proposes an algorithm to quantify CPL in spatial-time data.
However, these solutions are \textbf{limited} to a specific set of $\varepsilon$-LDP mechanisms that have \emph{only two levels of perturbation probabilities} (e.g., Generalised randomised response (GRR)~\cite{LDP_Frequent_Itemset_Mining}, Optimised unary encoding (OUE)~\cite{ldp_Frequency_Estimation}, etc.).  
Meanwhile, some studies have analysed and proposed approaches to measure CPL in $\varepsilon$-DP (central setting)~\cite{correlated_DP, correlated_data_dp_conventional_metrics_MI}. 
These methods generally cannot be applied to LDP because, in LDP, each user perturbs their data locally. 
This prevents a central aggregator from perturbing the dataset to handle correlations in the same way as DP. 
Additionally, in LDP, correlations often arise within a user's attributes rather than across different user records. 
Moreover, these studies commonly adopted conventional correlation metrics (e.g., Pearson correlation coefficients (PCC)~\cite{wikipedia_pearson}, covariance, etc.) in (L)DP analysis~\cite{LDP-PCC-dependency-graph-IoT, secon_Collecting_High-Dimensional_Correlation_LDP, correlated_DP}. For example,~\cite{secon_Collecting_High-Dimensional_Correlation_LDP} proposes a covariance-based algorithm to quantify CPL, and~\cite{LDP-PCC-dependency-graph-IoT} calibrates privacy budgets using PCC to trade-off privacy and utility. 


In summary, there has been \textbf{no} comprehensive investigation into CPL in LDP for real-world data and recent studies predominantly rely on conventional correlation metrics to measure the dependency (e.g., $D$ in  Fig.~\ref{fig:intro}(a)) to characterise CPL without a validation. 
Moreover, the potential information gain of an adversary on the impact of correlation within the context of generic ($\varepsilon,\delta$)-LDP mechanisms remains unexplored. 

To this end, this paper presents a comprehensive analysis of \emph{correlation-induced privacy leakage} (CPL) in LDP through practical experimentation and theoretical analysis, addressing the following key open research questions: 
\begin{itemize}
    \item \textit{\textbf{Q1.} What is the nature of CPL in real-world data?}
    \item \textit{\textbf{Q2.} Can we use conventional correlation metrics to quantify CPL?}
    \item \textit{\textbf{Q3.} How to systematically measure CPL under LDP?}
\end{itemize}

First, we statistically evaluate the impact on privacy due to the correlation using four real-world datasets in Section~\ref{section:motivation}. In the results (Section~\ref{section:experimental_results}), we observe that the worst-case privacy leakage scenarios occur relatively rarely, leaving substantial room to utilise both the privacy budget and population (in random-sampling methods) to enhance data utility without compromising privacy. Moreover, we observe that CPL is \textbf{not} necessarily symmetric. 
For instance, if $A$ and $B$ are two correlated attributes, CPL caused by $B$ about $A$ can differ from the CPL caused by $A$ about $B$. Additionally, we observe that in most cases, CPL \emph{saturates} as the privacy budget increases. 
These findings answer \textbf{\textit{Q1}}. 

To answer \textbf{\textit{Q2}}, next, we analyse the relationship between conventional correlation metrics and CPL in  Section~\ref{section:compare_the_observed_privacy_leakage_with_conventional_metrics}. 
We observe that although higher correlations leak more privacy, conventional metrics such as MI, PCC, etc., do \emph{not} accurately characterise CPL of the attributes. 
Moreover, these statistical results indicate a 5\% significance level\footnote{Statistical significance - A result is statistically significant at level $\alpha$ (i.e., 0.05) when, assuming the null hypothesis is true, the probability of observing data at least as extreme as the actual data is no greater than $\alpha$~\cite{moore2010basic}. 
} 
(i.e., $p < 0.05$).

Motivated by these statistical analyses, next, we theoretically analyse CPL of $(\varepsilon, \delta)$-LDP mechanisms in Section~\ref{section:methodology}.
To the best of our knowledge, this is the \textbf{first} CPL analysis for the generalised `approximated LDP' mechanism.
Then, we propose two novel algorithms to compute CPL. (i) Algorithm~\ref{algorithm_1} computes CPL when transition probabilities of the LDP mechanism and probability distribution between attributes are known. (ii) Algorithm~\ref{algorithm_2_part_2} computes CPL when $\varepsilon,\delta$, and probability distribution between attributes are known. 
Algorithm~\ref{algorithm_2_part_2} is useful to analyse CPL when transition probabilities are either unknown or sophisticated to compute (e.g., RAPPOR~\cite{rappor_original}, OLH~\cite{ldp_Frequency_Estimation}). 
These findings answer \textbf{\textit{Q3}}. Fig.~\ref{fig:intro}(b) summarises the proposed methods of computing $\textbf{CPL}$.

In Section~\ref{section:Theoretical Explanation for the Observations of Statistical Analysis}, we show that the discovered patterns in statistical analysis can be explained using the theoretical findings, which verifies these observations. 
Also, the derived theoretical results indicate that CPL could reach its \emph{maximum even under weak correlations}, which highlights the privacy risks of the usage of conventional correlation metrics (an example scenario is available in Appendix~\ref{section:appendix:Theoretical Explanation for the Observations of Statistical Analysis}).

Finally, in Section~\ref{section:Applications of CPL Analysis}, we demonstrate how these findings can be utilised to analyse and improve existing LDP solutions.
Here, we propose two novel benchmarks: (i) benchmarking CPL analysis algorithms, (ii) Utility vs CPL of LDP mechanisms.
Last but not least, we show how privacy budgets can be calibrated using publicly available information.

We summarise our \textbf{contributions} as follows:

\begin{enumerate}
    \item[$\bullet$] We statistically estimate correlation-induced privacy leakage (CPL) using four real-world datasets. 
    To the best of our knowledge, this is the \textbf{first empirical study} quantifying CPL in local differential privacy (LDP). 
    Notably, we discover that CPL between two correlated attributes is \textbf{asymmetric}, and that conventional correlation metrics \textbf{fail to characterise} CPL accurately.

    \item[$\bullet$] We provide a theoretical analysis of CPL under $(\varepsilon,\delta)$-LDP mechanisms, introducing two novel algorithms:  
    (i) Algorithm~\ref{algorithm_1}, which computes the exact CPL given the transition probabilities and attribute distributions; and  
    (ii) Algorithm~\ref{algorithm_2_part_2}, which computes a tight upper bound on CPL using only $\varepsilon$, $\delta$, and attribute distributions.  
    To our knowledge, this is the \textbf{first theoretical framework} for CPL analysis under $(\varepsilon,\delta)$-LDP.

    \item[$\bullet$] We propose \textbf{two novel benchmarks}:  
    (i) A CPL analysis benchmark to assess and compare existing privacy analysis algorithms; and  
    (ii) A privacy–utility benchmark that quantifies the trade-off between data utility and CPL exposure across LDP mechanisms.

    \item[$\bullet$] We validate our theoretical findings through experiments on both synthetic and real-world datasets.  
    Our results provide actionable insights for \textbf{privacy budget calibration} and demonstrate how CPL-aware strategies can improve the privacy–utility trade-off in practice.
\end{enumerate}

\section{Background}\label{section:background}

In this section, we present and review the related definitions and theorems required for the rest of the paper. First, we present the LDP definition and the related theorems. Next, we discuss the correlated attributes, and we define the different privacy leakages in multi-dimensional data. Finally, we define the threat model considered in this study.

\subsection{Local Differential Privacy}\label{subsection:local_differnetial_privacy}

This section reviews the LDP~\cite{LDP_duchi} definitions and related concepts. 
The \emph{approximated LDP} becomes \emph{pure LDP} when $\delta=0$. Here, $\delta \in [0,1)$ is the relaxation parameter, and a small $\delta$ guarantees better privacy.

\begin{definition}[($\varepsilon, \delta$)-Local Differential Privacy] \label{defn:LDP}
    Let the input alphabet be $\mathcal{X}$ and the output alphabet be $\mathcal{Y}$. A mechanism $Q$ is $(\varepsilon, \delta)$-locally differentially private if, 
    \begin{equation}
        \sup_{S\subseteq\mathcal{Y},x,x'\in\mathcal{X}} \frac{Q(S|x)-\delta}{Q(S|x')} \leq e^\varepsilon,
    \end{equation}
    where $Q(S|x) = p(Y=y_i, \ y_i \in S \ | \ X = x)$ is the privatisation mechanism. The function $p(\cdot)$ denotes the probability of an event. Here, $x \in \mathcal{X}$, $y_i \in \mathcal{Y} $, $\varepsilon \geq 0$ and $0 \leq \delta < 1$. When $\delta = 0$ this becomes $\varepsilon$-LDP.
\end{definition}

Moreover, we can define a variable $X$ as $\varepsilon$-local differentially private with respect to the release of $Y_1,\cdots, Y_n$, as stated in Definition~\ref{defn:LDP_multiple_attributes}.
According to the Definition~\ref{defn:LDP_multiple_attributes}, no adversaries can differentiate $x \in \mathcal{\Hat{X}}$ from $x' \in \mathcal{\Hat{X}}$ better than factor of $e^\varepsilon$ by observing/analysing any possible released $(y_1,\cdots,y_n)$ where $y_i \in \mathcal{Y}_i,\ i \in [n]$. 

\begin{definition}\label{defn:LDP_multiple_attributes}
    Let $X_1, \dots, X_n$ denote inputs of the mechanisms $Q_1, \dots, Q_n$ and $Y_1,\dots, Y_n$ represent the outputs (perturbed) of the mechanisms, respectively. Let $\mathcal{X}_1, \dots, \mathcal{X}_n$ be the alphabets of the inputs, and let $\mathcal{Y}_1,\dots, \mathcal{Y}_n$ be the alphabets of the outputs. Then, any input $\Hat{X}$ is $\varepsilon$-local differentially private against any release output of the system $(y_1,\dots,y_n)$ where $y_i \in \mathcal{Y}_i, \ i \in [n]$ if,
    \begin{equation}\label{defn:ldp_correlated}
        \underset{y_1 \in \mathcal{Y}_1,\dots,y_n \in \mathcal{Y}_n,x,x'\in\mathcal{\Hat{X}}}{\sup}\frac{p(Y_1 = y_1, \dots, Y_n = y_n|\Hat{X} = x)}{p(Y_1 = y_1, \dots, Y_n = y_n|\Hat{X} = x')} \leq e^\varepsilon.
    \end{equation}
The function $p(\cdot)$ denotes the probability of an event. Here, $\mathcal{\Hat{X}}$ is the alphabet of $\Hat{X}$.
\end{definition}

\subsection{Sequential Composition}

Sequential composition theorem (Theorem~\ref{thm:composition_sequential}) explains the privacy leakage when the same data is queried multiple times by (L)DP mechanisms~\cite{Algorithmic_Foundations_dwork}.

\begin{theorem}[Sequential Composition~\cite{LDP_survey_composition_theorem}]\label{thm:composition_sequential} 
    Let $\mathcal{M}_i(v)$ be an $(\varepsilon_i, \delta_i)$-LDP algorithm on an input value $v$, and $\mathcal{M}(v)$ be the sequential composition of $\mathcal{M}_1(v), \dots, \mathcal{M}_m(v)$. Then, $\mathcal{M}(v)$ satisfies $(\sum_{i=1}^m \varepsilon_i, \sum_{i=1}^m \delta_i)$-LDP.
\end{theorem}

\subsection{Correlated Attributes}\label{subsection:correlated_attributes}

In this work, the term ``correlation'' refers to \emph{both} linear and nonlinear dependencies. These correlations make privacy protection more challenging as correlated attributes could potentially reveal more information~\cite{no_free_lunch, correlated_DP}. The following sections formally define these privacy leakages.

\subsection{Privacy Leakages}\label{subsection:additional_privacy_leakage}
\begin{figure}[t]
  \centering
  \includegraphics[trim={4cm 0.45cm 1cm 0.8cm},clip,width=1.05\linewidth]{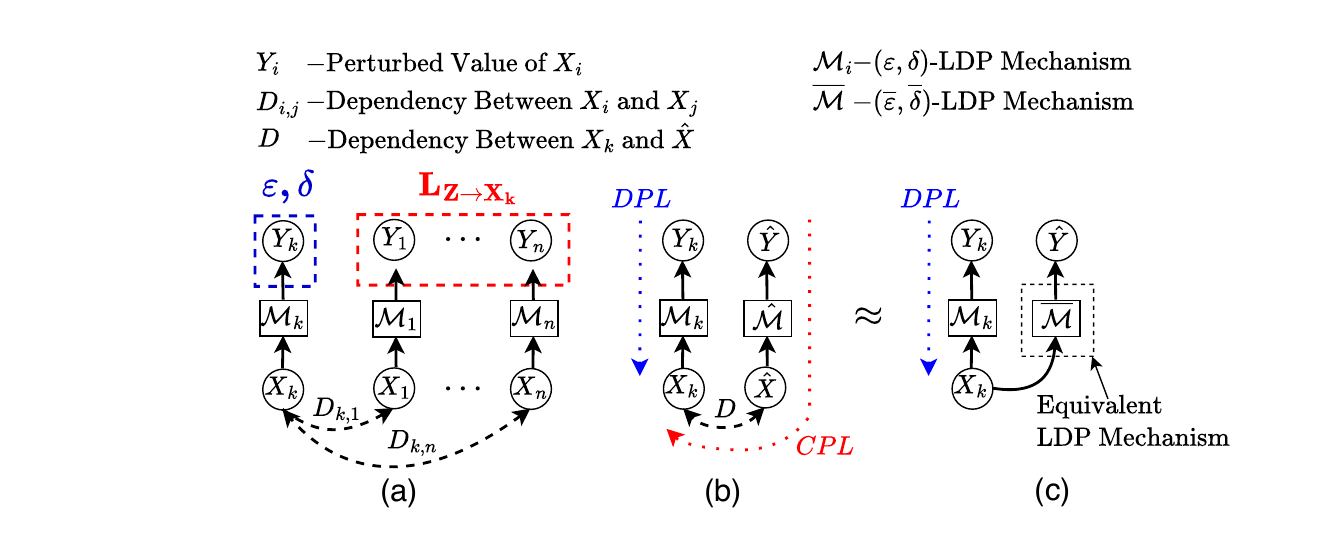}
  \caption{(a) Example dependency relationship of any attribute $X_k \in X$ with remaining attributes $X \setminus \{X_k\}$. Here, $\varepsilon$ represents the DPL of $X_k$ attribute and $L_{Z \rightarrow X_k}$ represents the CPL of $X_k$ attribute caused by $Z$ set of attributes, where $Z = \{X_1,\dots,X_n\} \setminus \{X_k\}$. (b) Example for privacy leakage analysis of $X_k$ attribute in simplified scenario. (c) Equivalent representation for CPL of $X_k$ in Fig.~\ref{fig:background_fig}(b).
  }
  \vspace{-1.5em}
  \label{fig:background_fig}
\end{figure}
Let $X_1, \dots, X_n \in X$ be $n$ correlated attributes and $Y_1, \dots, Y_n \in Y$ be perturbed versions obtained with $\mathcal{M}_1, \dots, \mathcal{M}_n \in \mathcal{M}$ $(\varepsilon, \delta)$-LDP mechanisms (Fig.~\ref{fig:background_fig}(a)). Let $\mathcal{X}_1, \dots, \mathcal{X}_n$ be the alphabets of inputs and $\mathcal{Y}_1, \dots, \mathcal{Y}_n$ be the alphabets of outputs, respectively. Then, we can define \emph{direct privacy leakage (DPL), correlation-induced privacy leakage (CPL), total correlation-induced privacy leakage (TCPL)}, and \emph{total privacy leakage (TPL)} terms as follows.

\subsubsection{Direct Privacy Leakage}

We define the \textbf{direct privacy leakage (DPL)} as the extent to which an adversary can infer information about a user's original data from its perturbed data (without correlated attributes' perturbed data). We can quantify the DPL of LDP mechanism from its privacy budget $\varepsilon$ and relaxation parameter $\delta$. Formally, let any attribute $X_k \in X$ be perturbed with $(\varepsilon, \delta)$-LDP mechanism $\mathcal{M}_k \in \mathcal{M}$. Then, the DPL of $X_k$ is at most $\varepsilon$ with relaxation of $\delta$.

For example, let us consider a simplified two-attribute scenario as shown in Fig.~\ref{fig:background_fig}(b). Let $X_k$ and $\hat{X}$ be two correlated attributes, and let $Y_k$ and $\hat{Y}$ be perturbed values of $X_k$ and $\hat{X}$ generated by LDP mechanisms $\mathcal{M}_k$ and $\mathcal{\hat{M}}$, respectively. Then, DPL of $X_k$ is caused by mechanism $\mathcal{M}_k$ due to the release of $Y_k$ as shown by the \textcolor{blue}{blue dotted arrow} in Fig.~\ref{fig:background_fig}(b).


\subsubsection{Correlation-Induced Privacy Leakage}

We define the \textbf{correlation-induced privacy leakage (CPL)} as the privacy leakage of an attribute that is caused by correlated neighbouring attributes. Formally, let $L_{Z \rightarrow X_k}$ denote the CPL experienced by any attribute $X_k \in X$ due to a subset of other correlated attributes $Z \subseteq X\setminus \{X_k\}$. Let $Y_k \in Y$ be the perturbed output of $X_k$, and $W \subseteq Y \setminus \{Y_k\}$ be the perturbed output of $Z$, and $\mathcal{W}$ be the output alphabet of $W$. Then, based on Definition~\ref{defn:LDP_multiple_attributes}, $L_{Z \rightarrow X_k}$ is given by, 
\begin{equation}\label{eqn:privacy_leakage_additional_xk}
    L_{Z \rightarrow X_k} = \ln \sup_{w\in \mathcal{W},x,x'\in\mathcal{X}_k} \frac{p(w|x)}{p(w|x')}.
\end{equation}

Referring to Fig.~\ref{fig:background_fig}(b), the CPL of $X_k$ from $\hat{X}$ is caused by mechanism $\hat{\mathcal{M}}$ due to the release of $\hat{Y}$ as shown by the \textcolor{red}{red dotted arrow}.
In Section~\ref{section:delta_greater_than_zero_theoretical_proof} (theoretical analysis), we show that CPL can be represented with a relaxation factor as well (e.g., $L_{Z \rightarrow X_k}$ can be represented using two components (i) privacy leakage component and (ii) relaxation component) when neighbouring attributes (i.e., $Z$) are perturbed with approximated LDP mechanisms ($\delta > 0$).
For example, suppose $L_{\hat{X} \rightarrow X_k}=(l_{\hat{X} \rightarrow X_k},\overline{f}^*_{\hat{X} \rightarrow X_k})$ (i.e., $l_{\hat{X} \rightarrow X_k}$ is privacy leakage component and $\overline{f}^*_{\hat{X} \rightarrow X_k}$ is relaxation component). Then, intuitively, CPL caused by $\hat{X}$ about $X_k$ has a similar effect as $X_k$ attribute is perturbed with ($\varepsilon=l_{\hat{X} \rightarrow X_k}, \delta=\overline{f}^*_{\hat{X} \rightarrow X_k}$)-LDP mechanism (i.e., $\overline{\mathcal{M}}$) as shown in Fig.~\ref{fig:background_fig}(c).


\subsubsection{Total Correlation-Induced Privacy Leakage}\label{section:Total Correlation-Induced Privacy Leakage}

We define the \textbf{total correlation-induced privacy leakage (TCPL)} as the summation of CPL caused by all the correlated neighbouring attributes in every attribute in the dataset. TCPL is given by,
\begin{equation}
    \operatorname{TCPL} = \sum_{x_k \in X} \sum_{x \in X\setminus \{x_k\}} L_{x \rightarrow x_k}.
\end{equation}
\subsubsection{Total Privacy Leakage}

We refer to \textbf{total privacy leakage (TPL)} as the privacy leakage of an attribute, including its DPL and CPL. Formally, let $L_{X_k}$ denote the TPL of any attribute $X_k \in X$. Then, $L_{X_k}$ can be calculated as
\begin{equation}\label{eqn:tpl_defn}
L_{X_k} = \ln \sup_{y_1\in \mathcal{Y}_1,\dots,y_n\in \mathcal{Y}_n,x,x'\in\mathcal{X}_k} \frac{p(y_1, \dots, y_n|x)}{p(y_1, \dots, y_n|x')},
\end{equation}
based on Definition~\ref{defn:LDP_multiple_attributes}.
Calculating the above equation is challenging as the number of attributes increases because of the curse of dimensionality. 
To overcome this, we can model TPL of $X_k$ attribute as a combination of DPL and CPL by neighbouring attributes.
Referring to Fig.~\ref{fig:background_fig}(b), privacy leakage of $X_k$ can be represented as two mechanisms $\mathcal{M}_k$ and $\overline{\mathcal{M}}$ as shown in Fig.~\ref{fig:background_fig}(c). Then, the upper bound of TPL of $X_k$ can be computed as the sum of privacy leakage of $\mathcal{M}_k$ and $\overline{\mathcal{M}}$ (i.e., DPL + CPL) using the sequential composition Theorem~\ref{thm:composition_sequential}.
For the generalised case, we can compute an upper bound for $L_{X_k}$, $\overline{L}_{X_k}$ using the sequential composition Theorem~\ref{thm:composition_sequential} as
\begin{equation}
\label{eqn:total_privacy_leakage_composition}
    \overline{L}_{X_k} = (\varepsilon + \sum_{x \in X \setminus \{X_k\}} l_{x \rightarrow X_k},\quad \delta + \sum_{x \in X \setminus \{X_k\}} \overline{f}^*_{x \rightarrow X_k}),
\end{equation}


where $\sum_{x \in X \setminus \{X_k\}} l_{x \rightarrow X_k}$, and $ \sum_{x \in X \setminus \{X_k\}} \overline{f}^*_{x \rightarrow X_k}$ are the sum of privacy leakage components and relaxation components of CPL, respectively. Here, $\varepsilon$ and $\delta$ denote the DPL and relaxation of $X_k$ under the $(\varepsilon,\delta)$-LDP mechanism. 
In a nutshell, we can compute an upper bound for the TPL of an attribute by summing up all the CPL and DPL. 

\subsection{Threat Model}

We consider a threat model wherein an adversary possesses distributional knowledge obtained from publicly available datasets, previously collected data, or, in the worst-case scenario, the ground-truth joint distribution. The adversary aims to perform \emph{attribute inference attacks}\footnote{Attribute inference attacks aim to infer a user's true input from the obfuscated output produced by a mechanism~\cite{distinguishability_attacks_2,PETS_statistical_audit}.} by leveraging statistical dependencies among attributes, even when local perturbation mechanisms are employed. Potential adversaries include honest-but-curious data collectors seeking to glean sensitive information from user data. In this context, adversarial knowledge gain is quantified using the LDP privacy leakage metric as detailed in Section~\ref{subsection:additional_privacy_leakage}.




\section{Statistical Analysis}\label{section:motivation}

In this section, we answer the first two research questions: \textbf{Q1} and \textbf{Q2}. Here, we statistically analyse the CPL using four real-world datasets. 
This section is organised as follows: Section~\ref{section:experimental_Setup_statistical} explains the experimental setup and Section~\ref{section:experiments_statistical} explains the conducted experiments. Next, we present results in Section~\ref{section:experimental_results}. Section~\ref{section:compare_the_observed_privacy_leakage_with_conventional_metrics} compares the observed privacy leakages with conventional correlation metrics. Finally, Section~\ref{section:limitations_of_the_statistical_evaluation_method} summarises the main statistical observations and discusses the challenges in the statistical method.

\subsection{Experimental Setup}\label{section:experimental_Setup_statistical}

\subsubsection{Environment}
We implement our experiments using Python 3.10 and run them on an AMD Ryzen Threadripper PRO 3955WX, 64GB RAM, and a 4TB storage computer. We have used two state-of-the-art LDP libraries~\cite{pureLDP2021, Arcolezi2022-multi-freq-ldpy}.

\subsubsection{Datasets}\label{section:datasets}
We use four real-world datasets. In these datasets, we have converted continuous numerical attributes into discrete attributes by binning, as this work focuses on privacy leakages related to discrete/categorical data.

\textbf{CelebA}~\cite{celebA_attributes} dataset contains 40 binary attributes of human faces with more than 200K samples. 

\textbf{Adult}~\cite{adult_dataset} is an income survey dataset containing 14 features and over 40K samples.

\textbf{CVD}~\cite{sulianova2021cardiovascular}
dataset consists of 11 features of 70K cardiovascular patient records (34,979 have cardiovascular
disease and 35,021 otherwise).

\textbf{DSS}~\cite{dss} dataset is released by the Department of Social Services (DSS) in Australia, and the selected dataset contains 27 features, with more than 38K samples.


\subsubsection{LDP Mechanisms}
In this experiment, we have used five commonly adopted LDP mechanisms: Generalised-randomised response (\textbf{GRR})~\cite{LDP_Frequent_Itemset_Mining}, \textbf{RAPPOR}~\cite{rappor_original}, Optimised unary encoding (\textbf{OUE})~\cite{ldp_Frequency_Estimation}, Optimised local hashing (\textbf{OLH})~\cite{ldp_Frequency_Estimation} and Exponential mechanism (\textbf{EXP}) ~\cite{Algorithmic_Foundations_dwork}, as many recent studies have used these mechanisms to develop more advanced mechanisms~\cite{PCKV_key_value}. We include more details in the Appendix~\ref{section:appendix:ldp_mechanisms} due to space limitations.

\subsubsection{Conventional Correlation Metrics}
We have used two mainly adopted correlation metrics in the literature: mutual information (\textbf{MI})~\cite{covers_mutual_information} and Pearson correlation coefficient (\textbf{PCC})~\cite{wikipedia_pearson}. 

\begin{figure}[t]
  \centering
  \includegraphics[trim={0.8cm 0.45cm 3cm 0cm},clip,width=1\linewidth]{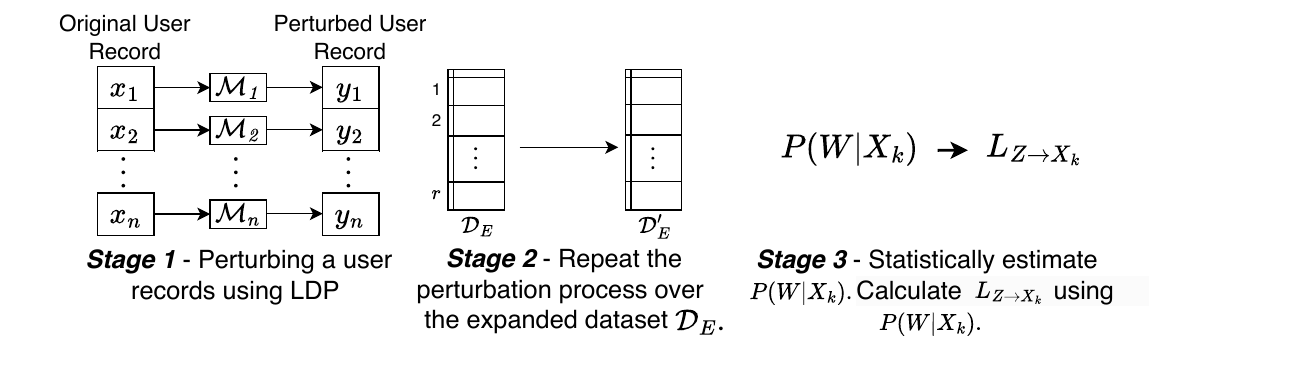}
  \caption{Workflow of the experiment. In Stage 1, we perturb each attribute ($x_1,\dots,x_n$) in a user record using LDP. This process should be repeated for all the records in the expanded dataset $\mathcal{D}_E$ and compute the perturbed dataset $D'_E$ as shown in Stage 2. In Stage 3, $P(W|X_k)$ is estimated using $D'_E$ and corresponding actual values. Here, $W = \{Y_1 \text{, }  \cdots \text{ ,} Y_n\} \setminus \{Y_k\}$  and $w \in \mathcal{W}$ where $\mathcal{W}$ is the alphabet of $W$ and $|\mathcal{W}| = l$.
  }
  \vspace{-1.5em}
  \label{fig:work_flow}
\end{figure}
\subsection{Experiments}\label{section:experiments_statistical}
The main objective of this experiment is to explore how privacy leakage behaves when attributes are correlated. As explained in Section~\ref{subsection:additional_privacy_leakage}, we can calculate the correlation-induced privacy leakage (CPL) (i.e., $L_{ Z\rightarrow X_k}$) caused by correlated attribute(s) using (\ref{eqn:privacy_leakage_additional_xk}). Here, we statistically estimate the value for this equation following a \emph{three-stage workflow} (Fig.~\ref{fig:work_flow}) to estimate CPL between correlated attributes. 

\begin{enumerate}
    \item  [$\bullet$]  \textbf{Stage~1} We apply LDP to perturb each user record in the dataset. Here, we independently perturb each attribute in the user records with the same privacy budget. 

    \item  [$\bullet$]  \textbf{Stage~2} We repeat this process over the entire dataset. 

    \item [$\bullet$] \textbf{Stage~3}, We compute the conditional probability distribution between correlated attributes to estimate CPL. For instance, if we require to calculate CPL of $X_k$ attribute caused by $Z= \{X_1,\dots,X_n\} \setminus \{X_k\}$ attributes, then we should calculate the probability distribution $P(W|X_k)$  where $W = \{Y_1,\dots, Y_n\} \setminus \{Y_k\}$ and $Y_i, i \in [n]$, are perturbed outputs of $X_i, i\in [n]$, respectively. Then, using $P(W|X_k)$, we can compute CPL using~\eqref{eqn:privacy_leakage_additional_xk}.
\end{enumerate}

Here, from an application perspective, \emph{Stage~1} represents the data perturbation at the user level, and \emph{Stage~2} represents the collection of that perturbed data by the data aggregator. 
To identify $Y_k$ from the perturbed output, we have adopted state-of-the-art attribute inference attack models from the literature~\cite{PETS_statistical_audit, distinguishability_attacks_2} (see Appendix~\ref{section:appendix:attacking_models} for details).
Further, a smaller sample count (user records) in the original dataset $\mathcal{D}$ can affect the estimation accuracy of $P(W|X_k)$. To overcome this, we augmented the dataset by replicating the records in $\mathcal{D}$, $r$ times and obtained the expanded dataset $\mathcal{D}_E$ for the perturbation. For example, if $\mathcal{D}$ has $N$ records, then $\mathcal{D}_E$ has $Nr$ records.
This increases the number of samples while preserving original relationships between attributes. In our experiments, we set $r=50$, which yields \textbf{statistically significant privacy leakage estimates} (at the $p<0.05$ level), even under a high privacy budget of $\varepsilon = 8$. 
To evaluate the \emph{statistical significance} of the estimated privacy leakage, we generate 10,000 surrogate datasets by independently permuting the values in each column (i.e., attribute) of the perturbed dataset. This preserves each attribute’s marginal distribution but breaks the correlation across attributes so we can determine whether the observed CPL is truly larger than any leakage we might see by randomly shuffling the data. 

We conduct the experiments in three main directions to evaluate different aspects: (i) CPL between attributes in the Adult dataset with different LDP mechanisms at $\varepsilon=1$. (ii) Privacy leakage between attributes over different privacy budgets in the CelebA dataset using GRR. (iii) CPL between each attribute in the CVD dataset using GRR at $\varepsilon=1$.

\subsection{Experimental Results}\label{section:experimental_results}

\begin{figure*}[h]
    \centering
    \includegraphics[trim={1cm 0.52cm 0cm .2cm},clip,width=1\linewidth]{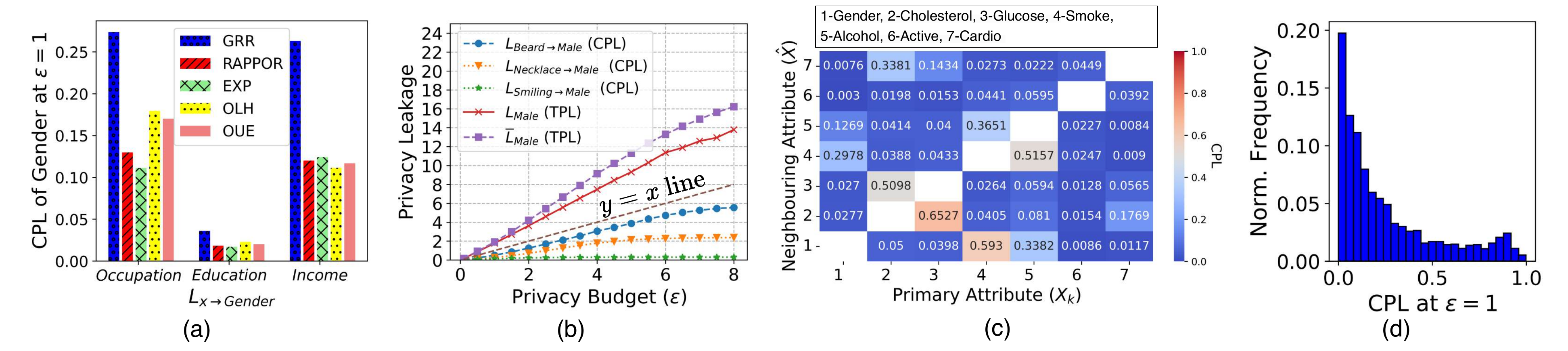}
    \caption{(a) CPL of \emph{Gender} attribute caused by three neighbouring attributes in Adult dataset at $\varepsilon = 1$. Here, data is perturbed with five different LDP mechanisms. (b) Privacy leakage of \emph{Male/Female} attribute caused by three other attributes in the CelebA dataset over a range of privacy budgets. Here, data is perturbed with GRR mechanisms. (c) The heatmap represents the CPL of each primary attribute $X_k$ caused by each neighbouring attribute $\Hat{X}$ (i.e. $L_{\Hat{X}\rightarrow X_k}$) in CVD dataset. Here, data is perturbed with GRR mechanisms. (d) Histogram of CPL between attributes in all four datasets perturbed with GRR mechanisms at $\varepsilon=1$. Here, the normalised frequency represents how many attributes (normalised) have each CPL value.}
    \vspace{-.5em}
\label{fig:motivation_results}
\end{figure*}

\subsubsection{CPL Between Attributes in Adult Dataset With Different LDP Mechanisms}
\label{section:CPL between Attributes in Adult Dataset with Different LDP Mechanisms}
As shown in Fig.~\ref{fig:motivation_results}(a), we evaluate the CPL of the gender attribute (\emph{Male/Female}) caused by three other attributes: \emph{Occupation}, \emph{Education} and \emph{Income} using GRR, RAPPOR, EXP, OLH and OUE  mechanisms. Here, we use  privacy budget $\varepsilon=1$ since, in practice, data is required to be perturbed with lower privacy budgets, like close to 1, for better privacy guarantee~\cite{LDP_survey_composition_theorem}. Out of five mechanisms, GRR leaks more privacy. Also note that although RAPPOR and EXP have relatively lower CPL, we cannot conclude that they are better over the other mechanisms since we only evaluate the CPL. Further, we observe that the five mechanisms exhibit consistent varying privacy leakage pattern across the attributes. For instance, knowing \emph{Occupation} data allows us to infer more about \emph{Gender} compared to the leakage caused knowing \emph{Education} as $L_{Occupation \rightarrow Gender} \gg L_{Education \rightarrow Gender}$. Since the GRR mechanism has the highest privacy leakage and the privacy leakage pattern is consistent across mechanisms, the rest of the experiments are done using GRR, as the GRR mechanism has the highest privacy leakage.

\subsubsection{Privacy Leakage Between Attributes Over Different Privacy Budgets in CelebA Dataset Using GRR} 

Here, we evaluate whether a person's gender could be inferred from other attributes using GRR mechanisms. For that, we selected three attributes from the CelebA dataset: \emph{Beard/No-Beard}, \emph{Wearing-necklace/No-wearing-necklace} and \emph{Smiling/No-smiling}, which have three different levels of correlation to gender from obvious biological knowledge. As expected, the \emph{Beard} attribute leaks more information about gender than the \emph{Necklace}, as shown in Fig.~\ref{fig:motivation_results}(b), because having a beard on a woman's face is \emph{very rare} compared to wearing a necklace by a man. Moreover, since \emph{Smiling} is a gender-independent expression, it does not reveal any significant information about the gender. We observe another interesting pattern in privacy leakage: It increases with privacy budget $\varepsilon$ for some point and becomes stable (saturates) afterwards. For example, $L_{Necklace \rightarrow Male}$ has maximum value of $\approx 2.4$ after $\varepsilon > 6$. These maximum CPL values increase with the correlation strength between attributes (i.e., $L_{Smiling \rightarrow Male} < L_{Necklace \rightarrow Male} < L_{Beard \rightarrow Male}$).
Moreover, we observe that CPL values are lower than the privacy budget (lie under the $y=x$ line). In our theoretical analysis, we prove that the CPL is always less than or equal to the privacy budget of neighbouring attributes' privacy budget (Corollary~\ref{corollary:maximum_additional_privacy_leakage_over_GG'} in Appendix~\ref{section:appendix:Extreme Scenarios}).
Next, we statistically evaluate the TPL of \emph{Male/Female} attribute ($L_{Male}$) with respect to the selected other three attributes using~\eqref{eqn:tpl_defn} (i.e., compute $P(W|X_k)$  where $W = \{\text{\emph{Beard/No-Beard}, \emph{Wearing-necklace/No-wearing-necklace}},\linebreak \text{\emph{Smiling/No-smiling}}\}$). Next, we compute $\overline{L}_{Male}$ using previously estimated CPLs from each attribute and~\eqref{eqn:total_privacy_leakage_composition}. As illustrates in Fig.~\ref{fig:motivation_results}(b), $\overline{L}_{Male}$ gives a tight upper bound for $L_{Male}$. Therefore, this result verifies the upper-bound approximation for TPL based on the sequential composition theorem in~\eqref{eqn:total_privacy_leakage_composition}.

\subsubsection{CPL Between Each Attribute in the CVD Dataset Using GRR at $\varepsilon=1$}\label{section:Additional Privacy Leakage between each Attribute in the Cardiovascular Dataset}

As shown in Fig.~\ref{fig:motivation_results}(c), we compute CPL between seven selected attributes from CVD dataset using GRR mechanisms at privacy budget $\varepsilon=1$. One of the main observations is that most leakages are \emph{close to zero} (the maximum value is 1 as we perturb this data with $\varepsilon=1$). To further validate this, we estimate the privacy leakage in the remaining four datasets, CelebA, Adult and DSS, as well. Next, we plot the histogram of the normalised frequencies of each CPL value as shown in Fig.~\ref{fig:motivation_results}(d). We can observe that most of the leakage values in all four datasets are \emph{close to zero}. Therefore, this implies that splitting the privacy budget equally among correlated attributes (SPL method) unnecessarily adds much redundant noise in real-world scenarios. Additionally, we note another interesting observation in Fig.~\ref{fig:motivation_results}(c): CPL between two attributes is \textbf{not necessarily symmetric}. For instance, if we know a person is a cardiovascular patient, then the patient's level of cholesterol is more likely to be revealed as $L_{Cardio \rightarrow Cholesterol} = 0.3381$. However, knowing the level of cholesterol, it is highly uncertain whether the person has cardiovascular disease or not as $L_{Cholesterol \rightarrow Cardio} = 0.1769$.

We can summarise the key observations of this statistical privacy leakage analysis as follows:
(i) Privacy leakage caused by correlated attributes is \emph{not} necessarily symmetric, (ii) Most of the time, CPL \textbf{saturates} as the privacy budget of the neighbouring attribute increases, and (iii) Most of the time, CPL is much less than the worst-case scenario.
Therefore, these findings answer the \textbf{Q1}, \emph{``What is the usual nature of CPL in real-world data?''}

Next, we compare in detail these privacy leakage results with conventional correlation metrics: MI and PCC. 

\subsection{Comparison of the Observed Privacy Leakage Results With Conventional Correlation Metrics}~\label{section:compare_the_observed_privacy_leakage_with_conventional_metrics}
\vspace{-2em}
\begin{figure}[h]
  \centering
  \includegraphics[trim={1cm 0.6cm 0cm .2cm},clip,width=0.95\linewidth]{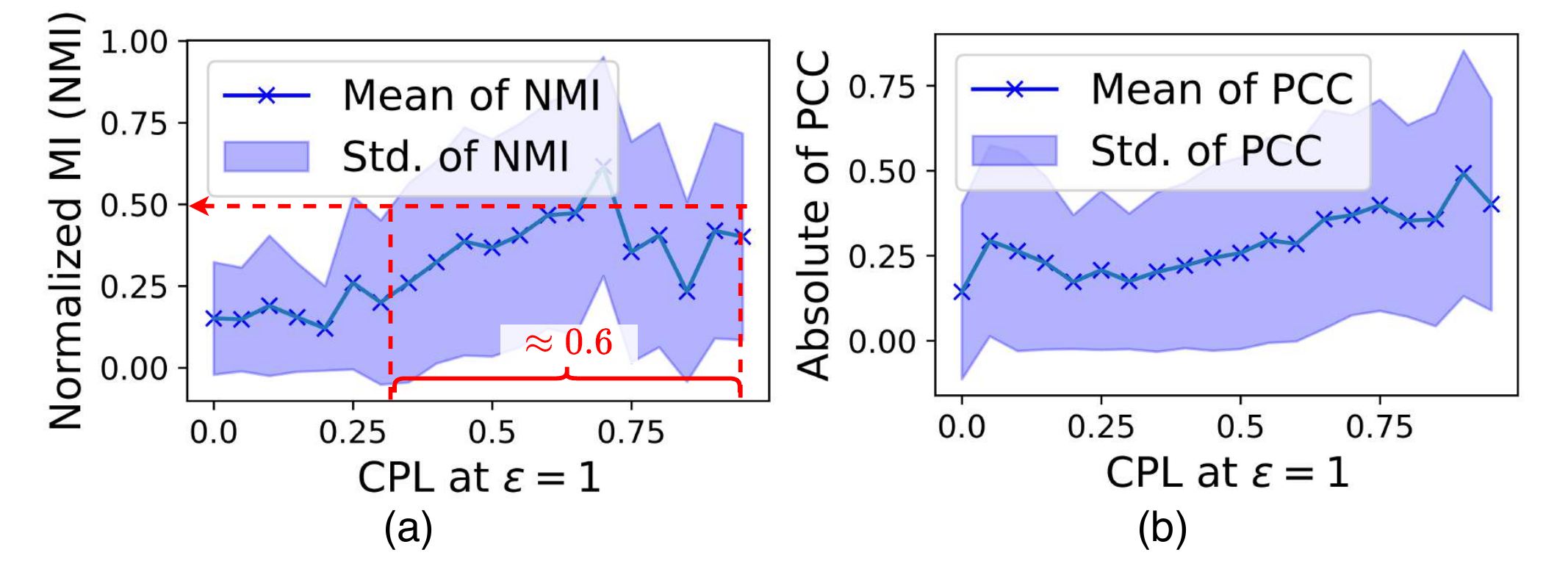}
  \caption{Mean and standard deviation (Std.) of correlation measured with conventional metrics vs CPL for all four datasets. (a) NMI vs CPL, (b) Absolute value of PCC vs CPL.} 
   \vspace{-1.5em}
   \label{fig:compare_conventional_metric_vs_privacy_leakage}
\end{figure}

Fig.~\ref{fig:compare_conventional_metric_vs_privacy_leakage} depicts the mean and standard deviation of the CPL versus correlation measured with conventional correlation metrics across all four datasets. We use GRR mechanisms with $\varepsilon=1$. As shown in Fig.~\ref{fig:compare_conventional_metric_vs_privacy_leakage}, CPL has a connection with the conventional correlation metrics. For example, CPL increases as normalised MI (NMI) increases. However, as the results show, CPL cannot be precisely predicted using NMI alone. For instance, if two attributes are correlated with NMI = 0.5, their CPL could be varied by approximately 0.6 as per Fig.~\ref{fig:compare_conventional_metric_vs_privacy_leakage}(a).

One of the primary reasons for the mismatch between conventional metrics and CPL could be the way these metrics are defined. For example, LDP is designed to evaluate the worst-case scenarios (by calculating supremum as in Definition.~\ref{defn:LDP}) of information leakage, whereas MI measures the average amount of information shared between variables (by calculating expectation~\cite{covers_mutual_information}).
Moreover, conventional correlation metrics give \emph{symmetric} measurements for the correlation, which \textbf{cannot} characterise the \emph{asymmetric} relationships in CPL (as shown in Section~\ref{section:Additional Privacy Leakage between each Attribute in the Cardiovascular Dataset}).
Therefore, this yields negative answers to \textbf{Q2} \emph{``Can conventional correlation metrics characterise CPL?''}

\subsection{Summary of the Statistical Observations and Challenges}\label{section:limitations_of_the_statistical_evaluation_method}

There are three main observations in the statistical analysis.
We have named them as \emph{OB1, OB2,} and \emph{OB3} as follows.

\begin{enumerate}
    \item[$\bullet$] \textbf{OB1} - Privacy leakage caused by correlated attributes is not necessarily symmetric.

    \item[$\bullet$] \textbf{OB2} - Most of the time, CPL saturates as the privacy budget of the neighbouring attribute increases.

    \item[$\bullet$] \textbf{OB3} - Conventional correlation metrics do not accurately characterise the CPL of LDP. 
\end{enumerate}

Additionally, this statistical method has \textbf{challenges} in terms of scalability. Specifically, statistical privacy leakage estimation requires a large number of samples and computational power to achieve statistically significant results. Therefore, in the next section, we derive a theoretical CPL to address these challenges. Moreover, we theoretically explain above statistical observations (i.e., \emph{OB1, OB2,} and \emph{OB3}) using derived theoretical results in Section~\ref{section:Theoretical Explanation for the Observations of Statistical Analysis}.

\section{Theoretical Analysis}\label{section:methodology}

This section theoretically analyses the correlation-induced privacy leakage (CPL) to answer \textbf{Q3}: \emph{``How can we systematically measure CPL under LDP?''}. 
In Section~\ref{section:Fomulate_the_Additional_Privacy_Leakage}, we formulate CPL for multiple attribute scenarios. Next, in Section~\ref {section:calculate additional privacy leakage between two correlated attributes}, we solve the derived formulation. In Section~\ref{section:Calculating CPL for specific mechanism} and Section~\ref{section:Calculating CPL for generic mechanisms}, we propose two algorithms to compute the CPL. 
Finally, in Section~\ref{section:Theoretical Explanation for the Observations of Statistical Analysis}, we theoretically explain the statistical observation in Section~\ref{section:motivation}.


\subsection{Formulation of CPL}\label{section:Fomulate_the_Additional_Privacy_Leakage}

Let us assume that ${X_1, \dots, X_n} \in X$ are $n$ correlated attributes that are perturbed using $(\varepsilon, \delta)$-LDP mechanisms $\mathcal{M}_1, \dots, \mathcal{M}_n \in \mathcal{M}$  and generating perturbed outputs ${Y_1, \dots, Y_n} \in Y$, respectively as shown in Fig.~\ref{fig:background_fig}(a). We can calculate the CPL of any attribute $X_k \in X$ caused by an attribute $\hat{X} \in X \setminus \{X_k\}$ as stated in Proposition~\ref{prop:total_privacy_leakage_from_probability}.

\begin{proposition}\label{prop:total_privacy_leakage_from_probability}
    Let ${X_1, \dots, X_n} \in X$ be $n$ correlated attributes that are perturbed using ($\varepsilon$,$\delta$)-LDP mechanisms $\mathcal{M}_1,\dots, \mathcal{M}_n \in \mathcal{M}$ producing perturbed outputs ${Y_1, \dots, Y_n} \in Y$, respectively. Let $\mathcal{X}_1, \dots, \mathcal{X}_n$ be the alphabets of the inputs and $\mathcal{Y}_1, \dots, \mathcal{Y}_n$ be the alphabets of the outputs.

    Suppose $L_{\Hat{X} \rightarrow X_k}$ is the correlation-induced privacy leakage (CPL) of any attribute $X_k \in X$ caused by release of attribute $\hat{X} \in X\setminus \{X_k\}$, then it is given by,
\begin{equation}\label{eqn:expanded_additional_privacy_leakage}
        L_{\Hat{X} \rightarrow X_k} = \ln{\sup_{y\in \mathcal{\Hat{Y}},x,x'\in\mathcal{X}_k} \left(\frac{\sum_{u \in \mathcal{\Hat{X}}} p( y|u)p(u|x)}{\sum_{v \in \mathcal{\Hat{X}}} p( y|v)p(v|x')}\right)}.
\end{equation}
    Here, $\hat{Y}$ is perturbed value of $\hat{X}$, and $\mathcal{\hat{Y}}$ is the alphabet of $\hat{Y}$.
\end{proposition}

The proof of Proposition~\ref{prop:total_privacy_leakage_from_probability} is available in Appendix~\ref{proof:prop:total_privacy_leakage_from_probability}.

Then, we can calculate the total privacy leakage (TPL) of any attribute $X_k$ if we have all the CPLs caused by each neighbouring attribute (i.e., $L_{X_i \rightarrow X_k}, i \in [n]\setminus \{k\}$) using~\eqref{eqn:total_privacy_leakage_composition}. Therefore, next, we calculate CPL between two correlated attributes.

\subsection{Calculating CPL Between Two Correlated Attributes}\label{section:calculate additional privacy leakage between two correlated attributes}

Suppose $\mathcal{\Hat{X}} = \{a_1, \dots, a_t\}$ is the alphabet of the attribute $\Hat{X}$ with $|\mathcal{\Hat{X}}|=t$. Then, we can define terms $p(y|u), p(u|x), p(u|x'), \ u \in \mathcal{\Hat{X}},\ y \in \mathcal{\Hat{Y}},\  x,x' \in \mathcal{X}_k$ in vector forms as $C = (p(y|a_1),\dots, p(y|a_t)),
        G = (p(a_1|x), \dots,p(a_t|x)),
        G'= (p(a_1|x'), \dots,p(a_t|x'))$.
Then, we can write (\ref{eqn:expanded_additional_privacy_leakage}) as shown in (\ref{eqn:exact_solution_for_additional_privacy_leakage}) with removing $\ln$. 
\begin{equation}\label{eqn:exact_solution_for_additional_privacy_leakage}
    e^{L_{\Hat{X} \rightarrow X_k}} = \sup_{y\in \mathcal{\Hat{Y}},x,x'\in\mathcal{X}_k} \frac{C^TG}{C^TG'}.
\end{equation}
Here, $T$ denotes the transpose operation. Therefore, we can calculate CPL of any attribute $X_k \in X$ caused by any other correlated attribute $\Hat{X}\in X \setminus \{X_k\}$ (i.e., $L_{\Hat{X} \rightarrow X_k}$), when we have the transition probabilities of the LDP mechanism (i.e., $C$) and the conditional probability distribution $P(\Hat{X}|X_k)$ (i.e., $G$ and $G'$). 
However, sometimes the LDP mechanisms are characterised using a privacy budget $\varepsilon$ and relaxation parameter $\delta$, as it could be challenging to obtain transition probabilities in sophisticated LDP mechanisms such as RAPPOR~\cite{rappor_original}, OLH~\cite{ldp_Frequency_Estimation}, etc. Therefore, next, we analyse the CPL in two different scenarios: (i) assuming we know the transition probabilities of the LDP mechanism in Section~\ref{section:Calculating CPL for specific mechanism} (covers simple mechanisms such as GRR and EXP), and (ii) assuming we only have $\varepsilon, \delta$ of LDP mechanisms in Section~\ref{section:Calculating CPL for generic mechanisms} (generic solution useful for analysing sophisticated LDP mechanisms.).  

\subsection{Calculating CPL for Known $G$, $G'$ and Transition Probabilities $C$}\label{section:Calculating CPL for specific mechanism}

When we have $G$, $G'$ and $C$, we can compute the \textbf{actual CPL} using
Algorithm~\ref{algorithm_1}.
Here, we find the maximum CPL, $L_{\hat{X}\rightarrow X_k}$ by traversing through all possible $C, G$ and $G'$ values over $\hat{X}$ and $X_k$ alphabet in for-loops at lines 2 and 4, using exhaustive search. 
In line 7, we keep the maximum value of $L_{\hat{X}\rightarrow X_k}$.

\textbf{Time Complexity Analysis - } Lines~2 and~4 execute nested \emph{for-loops} over $m$ and $t$, with each iteration performing a vector multiplication of cost $O(t)$, yielding $O(m^2 t^2)$ overall, where $m$ and $t$ denote the alphabet sizes of $X_k$ and $\hat{X}$, respectively. Therefore, the time complexity of Algorithm~\ref{algorithm_1} is $O(m^2 t^2)$.

Next, we validate the Algorithm~\ref{algorithm_1} results using statistically estimated CPL for GRR and EXP mechanisms as summarised in Appendix~\ref{section:Results: Statistical vs Theoretical algo 1}. According to the results, the normalised mean square error (NMSE) between theoretical (Algorithm~\ref{algorithm_1}) and statistically estimated CPL is always marginal ($<5.5 \times 10^{-3}$), which validates the Algorithm~\ref{algorithm_1}.

\begin{algorithm}[h]
\caption{Compute \emph{Actual CPL} caused by $\Hat{X}$ attribute about $X_k$ attribute for known $G,G'$ and $C$.}\label{algorithm_1}
\begin{algorithmic}[1]
    \Statex \textbf{Input:}
    \Statex \hspace{1em} $P(\Hat{X}|X_k)$ \Comment{\textit{Conditional probability distribution.}}
    \Statex \hspace{1em} $P(\hat{Y}|\hat{X})$\Comment{\textit{Transition probabilities of mechanism $\hat{M}$.}}
    \State $l \gets 0$ \Comment{\textit{Variable to keep the maximum leakage.}}
    \ForAll{$\hat{x} \in \mathcal{\hat{X}}$} 
    \State$ C \gets P(\hat{Y}|\hat{X}=\hat{x})$ \Comment{\textit{$C$ vector}}
    \ForAll{$x,x' \in \mathcal{X}_k, \quad x \neq x'$}
        \State$ G \gets P(\Hat{X}|X_k=x)$ \Comment{\textit{$G$ vector}}
        \State$ G' \gets P(\Hat{X}|X_k=x')$\Comment{\textit{$G'$ vector.}}
        \State $ l \gets \max (l, \ln\frac{C^T G}{C^T G'})$\Comment{\textit{Compute APL using~\eqref{eqn:exact_solution_for_additional_privacy_leakage}.}}
    \EndFor
    \EndFor
    \State \Return $L_{\Hat{X} \rightarrow X_k} = l$
\end{algorithmic}
\end{algorithm}


\subsection{Calculating CPL for Known $G$, $G',\varepsilon$ and $\delta$}\label{section:Calculating CPL for generic mechanisms}

We formulate the CPL computation as a maximisation problem as shown in (\ref{eqn:optimization_problem_over_C}). $G$ and $G'$ are two probability vectors corresponding to any two alphabet values of $X_k$.
\begin{equation} \label{eqn:optimization_problem_over_C}
    \begin{split}
        \underset{C}{\textbf{maximise}} & \quad F(C,G,G') = \frac{C^TG}{C^TG'}   \\
        \textbf{subject to} 
        &\quad c \leq e^{\varepsilon}\Tilde{c}+\delta, \ \forall c, \Tilde{c} \in C \ \textit{($\varepsilon,\delta$)-LDP constraints,}\\
        &\quad 0< c \leq 1, \ \forall c \in C \quad \quad \textit{probability constraints.}
    \end{split}
\end{equation} 
We can show that the maximum value of $F(C,G,G')$ is given by the boundary values of $C$ as stated in Proposition~\ref{prop:c_max c_min}.

\begin{proposition}\label{prop:c_max c_min}
Let $C$ be a size $t$ vector and let $c_{max} = \max(C)$ and $c_{min} = \min(C)$. Then, consider a vector $\Tilde{C}$ which consists only the elements $c_{min},c_{max}$ (i.e., $\Tilde{C} \in \{c_{\min}, c_{\max}\}^t$). Then, for a given function $F(C, G, G')$, we have $F(\Tilde{C},G,G') \geq F(C, G, G')$. Therefore, it is sufficient to consider $\Tilde{C}$ to calculate the maximum value of $F(C, G, G')$. Here, $c_{max}$ and $c_{min}$ satisfy $c_{max} = e^{\varepsilon}c_{min}+\delta$ at the global maximum of $J(S,c_{min}, c_{max}, G, G')$.
\end{proposition}
The proof of Proposition~\ref{prop:c_max c_min} is available in Appendix~\ref{proof:prop:c_max c_min}.



Then, we can express $F(C,G,G')$ by separating the elements into $c_{min}$ and $c_{max}$ coefficients using Proposition.~\ref{prop:c_max c_min}. Let $S \subseteq [t]$ denote the indices of the elements associated with $c_{max}$. Based on Proposition.~\ref{prop:c_max c_min} $F(C,G,G')$ can be written in the form of $H(c_{min},S, G,G')$ as
\begin{equation}\label{eqn:H_with_delta}
    H(\cdot) = \frac{c_{min}\sum_{i \in [t]\setminus S}g_i+(e^{\varepsilon}c_{min}+\delta)\sum_{j \in S}g_j}{c_{min}\sum_{i \in [t]\setminus S}g'_i+(e^{\varepsilon}c_{min}+\delta)\sum_{j \in S}g'_j}.
\end{equation}
Therefore, we can solve \eqref{eqn:optimization_problem_over_C} by calculating the maximum $H(\cdot)$ over $S$. Let $H^*$ be $H^*=\max_S H(\cdot)$. Then, we can compute CPL (i.e., $L_{\Hat{X} \rightarrow X_k}$) by considering the maximum $H^*$ over all possible $G,G'$ corresponding to the alphabets of $X_k$  as
\begin{equation}\label{eqn:H* and privacy leakage}
    L_{\Hat{X} \rightarrow X_k} = \ln \max_{x,x' \in \mathcal{X}_k} H^*.
\end{equation}
Therefore, next, we calculate $H^*$ in two different cases: $\delta > 0$ and $\delta = 0$.

\subsubsection{Case $\delta > 0$}\label{section:delta_greater_than_zero_theoretical_proof}

$c_{\min}$ depends on the transition probabilities of the LDP mechanism. Its exact value is unknown, as our objective here is to compute CPL using $\varepsilon$ and $\delta$.
Therefore, we \emph{remove} $c_{min}$ from the expression. Here, we know that \eqref{eqn:H_with_delta} is monotone over $c_{min}$ and $0 < c_{min} \leq 1 < \infty$. Therefore, we have,
\begin{equation}\label{eqn:H_upperbounds_case1_delta_greater_than_0}
    H(\cdot) < \max \left\{\frac{\sum_{i \in [t] \setminus S}g_i+e^{\varepsilon}\sum_{j \in S}g_j}{\sum_{i \in [t] \setminus S}g'_i+e^{\varepsilon}\sum_{j \in S}g'_j}, \frac{\sum_{j \in S} g_j}{\sum_{j \in S} g'_j} \right\}.
\end{equation}
Here, $\frac{\sum_{i \in [t] \setminus S}g_i+e^{\varepsilon}\sum_{j \in S}g_j}{\sum_{i \in [t] \setminus S}g'_i+e^{\varepsilon}\sum_{j \in S}g'_j}$ is monotone over $e^\varepsilon$, $1 \leq e^\varepsilon \leq \infty$. Therefore, $\frac{\sum_{i \in [t] \setminus S}g_i+e^{\varepsilon}\sum_{j \in S}g_j}{\sum_{i \in [t] \setminus S}g'_i+e^{\varepsilon}\sum_{j \in S}g'_j} \leq \linebreak \max\left\{1, \frac{\sum_{j \in S} g_j}{\sum_{j \in S} g'_j}\right\}$. Thus, from \eqref{eqn:H_upperbounds_case1_delta_greater_than_0}, $H^* < \max_S \left\{1, \frac{\sum_{j \in S} g_j}{\sum_{j \in S} g'_j} \right\}$. We know $\max_S \frac{\sum_{j \in S} g_j}{\sum_{j \in S} g'_j} = \max_{i \in [t]}\frac{g_i}{g'_i}$ from Lemma~\ref{lemma:maximum_value_of_linear_fraction}. Moreover, we know $1 \leq \max_{i \in [t]}\frac{g_i}{g'_i}$. Therefore, we have $H^* < \max_{i \in [t]}\frac{g_i}{g'_i}$. From \eqref{eqn:H* and privacy leakage}, we have the CPL $L_{\Hat{X} \rightarrow X_k} \leq \ln  \max_{\Hat{x}\in \mathcal{\Hat{X}},x,x'\in\mathcal{X}_k} \frac{p(\Hat{x}|x)}{p(\Hat{x}|x')} $ for $\delta >0$ scenario. However, as stated in Corollary~\ref{corollary:max_privacy_leakage_privacy_budget}, the maximum $L_{\Hat{X}\rightarrow X_k}$ as $\varepsilon$ varies is $\max_\varepsilon L_{\Hat{X}\rightarrow X_k} = \ln \max_{\Hat{x}\in \mathcal{\Hat{X}},x,x'\in\mathcal{X}_k} \frac{p(\Hat{x}|x)}{p(\Hat{x}|x')}$. This implies that when $\delta > 0$, we \textbf{cannot} calculate a tighter upper bound than the maximum CPL without knowing the $c_{min}$ of the LDP mechanism.

\begin{lemma}\label{lemma:maximum_value_of_linear_fraction}
    Let $G=(g_1, g_2,\dots,g_t)$ and $G'=(g'_1, g'_2,\dots,g'_t)$ be two vectors of size $t$ where $0 \leq g_i \leq 1$ and $0 \leq g'_i \leq 1$ for all $i \in [t]$. Suppose $S \subseteq [t]$ be a nonempty subset. Then, $\max_{S} \frac{\sum_{i \in S} g_i}{\sum_{i \in S} g'_i} = \max_{i \in [t]} \frac{g_i}{g'_i}$.
\end{lemma}

The proof of Lemma~\ref{lemma:maximum_value_of_linear_fraction} is available in Appendix~\ref{proof:lemma:maximum_value_of_linear_fraction}.

\begin{corollary}\label{corollary:max_privacy_leakage_privacy_budget}
    Let $X_k$ and $\Hat{X}$ be two correlated attributes. Suppose $\Hat{X}$ attribute is perturbed by ($\varepsilon,\delta$)-LDP mechanism $\Hat{M}$. Then, the CPL caused by $\Hat{X}$ about $X_k$ saturates as $\varepsilon$ increases, as shown in \eqref{eqn:corollary:max_privacy_leakage_privacy_budget}. Moreover, this is the maximum $L_{\Hat{X}\rightarrow X_k}$ as $\varepsilon$ varies (i.e., $\max_{\varepsilon} L_{\Hat{X}\rightarrow X_k}$),
    \begin{equation}\label{eqn:corollary:max_privacy_leakage_privacy_budget}
        \lim_{\varepsilon \rightarrow \infty} L_{\Hat{X}\rightarrow X_k} = \max_{\varepsilon} L_{\Hat{X}\rightarrow X_k} = \ln \max_{\Hat{x}\in \mathcal{\Hat{X}},x,x'\in\mathcal{X}_k} \frac{p(\Hat{x}|x)}{p(\Hat{x}|x')}.
    \end{equation}
    Here, $\mathcal{\Hat{X}}$ and $\mathcal{X}_k$ are alphabets of  $\Hat{X}$ and $X_k$, respectively. 
\end{corollary}
The proof of Corollary~\ref{corollary:max_privacy_leakage_privacy_budget} is available in Appendix~\ref{proof:corollary:max_privacy_leakage_privacy_budget}.

However, we identify that we can reformulate the $H(\cdot)$ in a form of relaxed privacy leakage similar to the approximated LDP. From \eqref{eqn:H_with_delta},
\begin{equation}
\begin{split}
    H(\cdot) = & \frac{c_{min}\sum_{i \in [t] \setminus S}g_i+e^{\varepsilon}c_{min}\sum_{j \in S}g_j}{c_{min}\sum_{i \in [t] \setminus S}g'_i+(e^{\varepsilon}c_{min}+\delta)\sum_{j \in S}g'_j} \\
    & + \frac{\delta\sum_{j \in S}g_j}{c_{min}\sum_{i \in [t] \setminus S}g'_i+(e^{\varepsilon}c_{min}+\delta)\sum_{j \in S}g'_j}. \\
\end{split}
\end{equation}

\noindent Since $\delta >0 $, we have $\frac{c_{min}\sum_{i \in [t] \setminus S}g_i+e^{\varepsilon}c_{min}\sum_{j \in S}g_j}{c_{min}\sum_{i \in [t] \setminus S}g'_i+(e^{\varepsilon}c_{min}+\delta)\sum_{j \in S}g'_j} <  \frac{c_{min}\sum_{i \in [t] \setminus S}g_i+e^{\varepsilon}c_{min}\sum_{j \in S}g_j}{c_{min}\sum_{i \in [t] \setminus S}g'_i+e^{\varepsilon}c_{min}\sum_{j \in S}g'_j}$. Moreover, from \eqref{eqn:privacy_leakage_additional_xk}, we know $p(y|x') = c_{min}\sum_{i \in [t] \setminus S}g'_i+(e^{\varepsilon}c_{min}+\delta)\sum_{j \in S}g'_j$ where $y \in \Hat{\mathcal{Y}},\ x,x' \in \mathcal{X}_k$. Therefore, we have, 

\begin{equation}\label{eqn:H_delta_greater than zero}
    H(\cdot) = \frac{p(y|x)}{p(y|x')}<\frac{\sum_{i \in [t] \setminus S}g_i+e^{\varepsilon}\sum_{j \in S}g_j}{\sum_{i \in [t] \setminus S}g'_i+e^{\varepsilon}\sum_{j \in S}g'_j} + \frac{\delta\sum_{j \in S} g_j}{p(y|x')}.
\end{equation}

Above \eqref{eqn:H_delta_greater than zero} has a form similar to the \textbf{approximated LDP} definition  (Definition.~\ref{defn:LDP}) as
\begin{equation}\label{eqn:reformed_ldp_def}
    \frac{p(y|x)}{p(y|x')} \leq e^\varepsilon + \frac{\delta}{p(y|x')}, \quad \forall y \in \mathcal{Y}, \ \forall x,x' \in \mathcal{X},
\end{equation}
where $\mathcal{X}$ and $\mathcal{Y}$ are the alphabets of inputs and outputs of the ($\varepsilon,\delta$)-LDP mechanism, respectively.

Therefore, we can take $\frac{\sum_{i \in [t] \setminus S^*}g_i+e^{\varepsilon}\sum_{j \in S^*}g_j}{\sum_{i \in [t] \setminus S^*}g'_i+e^{\varepsilon}\sum_{j \in S^*}g'_j}$ as the \textbf{privacy leakage} and $\delta\sum_{j \in S^*} g_j$ as the \textbf{relaxation amount} of the privacy leakage guarantee. Here, $S^*$ is the optimal $S$ which maximise \linebreak $\frac{\sum_{i \in [t] \setminus S}g_i+e^{\varepsilon}\sum_{j \in S}g_j}{\sum_{i \in [t] \setminus S}g'_i+e^{\varepsilon}\sum_{j \in S}g'_j}$. Next, we show that $\delta=0$ scenario also has same privacy leakage $\max_S \frac{\sum_{i \in [t] \setminus S}g_i+e^{\varepsilon}\sum_{j \in S}g_j}{\sum_{i \in [t] \setminus S}g'_i+e^{\varepsilon}\sum_{j \in S}g'_j}$.

\subsubsection{Case $\delta = 0$}

Then, we have,
\begin{equation}\label{eqn:H}
    H(r,G,G', \varepsilon, \delta = 0) = \frac{\sum_{i \in [t] \setminus S}g_i+e^{\varepsilon}\sum_{j \in S}g_j}{\sum_{i \in [t] \setminus S}g'_i+e^{\varepsilon}\sum_{j \in S}g'_j}.
\end{equation}
Since the privacy leakage of both $\delta >0$ and $\delta =0$ have the same expression given in \eqref{eqn:H}, we can use the same algorithm to solve both scenarios. Here, note that when $\delta>0$, there is a relaxation component for the privacy leakage. Next, let us define terms $\overline{H}, H^*$ and $f^*$ for ease of usage.
\begin{align}
    \overline{H} &= \frac{\sum_{i \in [t] \setminus S}g_i+e^{\varepsilon}\sum_{j \in S}g_j}{\sum_{i \in [t] \setminus S}g'_i+e^{\varepsilon}\sum_{j \in S}g'_j}, \\
    \overline{H}^* &= \max_S \overline{H}, \\
    f^* &= \delta\sum_{j \in S^*} g_j, \text{ where } S^* = \arg\max_S \overline{H}.
\end{align}
Therefore, CPL (i.e., $L_{\hat{X} \rightarrow X_k}$) can be represented using two components as $L_{\hat{X} \rightarrow X_k} = (\ln \max_{x,x' \in X_k} \overline{H}^*,\overline{f}^*)$. Here, we take the relaxation component $\overline{f}^*$ as the $f^*$ value corresponds to the maximum $\overline{H}^*$. If $\delta=0$, we can write CPL as $L_{\hat{X} \rightarrow X_k} = \ln \max_{x,x' \in X_k} \overline{H}^*$.

We notice that this problem can be solved time efficiently $O(m^2t\log t)$ by adopting a greedy strategy, as explained in the next section. 

\subsubsection{Algorithm of Calculating CPL for Known $G$, $G',\varepsilon$ and $\delta$}\label{section:Proposed Algorithm to Compute CPL}


\begin{algorithm}[h]
\caption{Compute an \emph{Upper Bound for CPL} caused by $\Hat{X}$ attribute about $X_k$ attribute (i.e., $L_{\Hat{X} \rightarrow X_k}$).}\label{algorithm_2_part_2}
\begin{algorithmic}[1]
    \Statex \textbf{Input:}
    \Statex \hspace{1em} $P(\Hat{X}|X_k)$ \Comment{\textit{Conditional probability distribution.}}
    \Statex \hspace{1em} $\varepsilon, \delta$\Comment{\textit{Privacy budget and relaxation of $\Hat{M}$.}}
    \State $l \gets 0$ \Comment{\textit{Variable to keep the maximum leakage.}}
    \ForAll{$x,x' \in \mathcal{X}_k, \quad x \neq x'$}     \State$ G \gets P(\Hat{X}|X_k=x)$ \Comment{\textit{$G$ vector}}
        \State$ G' \gets P(\Hat{X}|X_k=x')$\Comment{\textit{$G'$ vector.}}
        \State $Q \gets \text{sort}(G \oslash G')$ \Comment{\textit{Descending order of element-wise ratios.}}
    \State $A, B \gets 0$ \Comment{\textit{Variables to keep the summations.}}
    \ForAll{$q \in Q$}
            \If{$q \geq \frac{1+A(exp(\varepsilon)-1)}{1+B(exp(\varepsilon)-1)}$}
                \State $A \gets A + g$ \Comment{\textit{Element of $G$ corresponds to $q$.}}
                \State $B \gets B + g'$ \Comment{\textit{Element of $G'$ corresponds to $q$.}}
            \EndIf
        \EndFor
    \State $\overline{H}^* =\frac{1+A(exp(\varepsilon)-1)}{1+B(exp(\varepsilon)-1)}$, $f^* = \delta A$
        \If{$l < \ln\overline{H}^*$}
                \State $l \gets \ln\overline{H}^*$ \Comment{\textit{Keep the maximum leakage.}}
                \State $\overline{f}^* \gets f^*$ \Comment{\textit{Keep the flexibility.}
                }
        \EndIf
    \EndFor
    
    \State \Return $L_{\Hat{X} \rightarrow X_k} = (l,\overline{f}^*)$ 
\end{algorithmic}
\end{algorithm}


Algorithm~\ref{algorithm_2_part_2} is an exhaustive search to find the maximum $\overline{H}^*$ by traversing through all possible $G$ and $G'$ values over $X_k$ alphabet in \emph{for-loop} at line 2. 
Next, we take two variables $A$ and $B$ to store $\sum_{j \in S}g_j$ and $\sum_{j \in S}g'_j$, respectively at lines 6. 
According to the Lemma.~\ref{lemma:sigma_greater_than_ratio}, to maximise $H(\cdot)$ for selected $G,G'$, we should add $g,g'$ to $A,B$, respectively, if $\frac{g}{g'} \geq \frac{1+A(exp(\varepsilon)-1)}{1+B(exp(\varepsilon)-1)}$. Therefore, we sort the element-wise ratio $G\oslash G'$ in descending order, and update $A$ and $B$ by checking the ratio of $\frac{g}{g'}$ as in lines 5-12.
According to Theorem~\ref{Theorem:optimality_of_algo}, this algorithm always finds the optimal solution for $\overline{H}$ for selected $G,G'$.
CPL $L_{\Hat{X} \rightarrow X_k}$ will be $\ln$ of the maximum $\overline{H}^*$ out of all $G,G'$, which will be stored in variable $l$ (lines 14-15). Also, we can take the $\overline{f}^*$ as the value of $f^*$ corresponds to the maximum $\overline{H}^*$ as shown in line 16. 

\textbf{Time Complexity Analysis - }
Time complexity for computing optimal $L_{\Hat{X} \rightarrow X_k}$ is $O(m^2t \log{t})$ since \emph{for-loops} at line 2 have  $O(m^2)$ time complexity and computing $\overline{H}^*$ for selected $G,G'$ (lines 2-17) has time complexity of $O(t \log{t})$ (because of the sorting $G \oslash G'$ step~\cite{WikipediaSortingAlgorithm}). Here, $m$ and $t$ are the alphabet sizes of $X_k$ and $\hat{X}$.

For the completeness of the theoretical analysis, Appendix~\ref{section:appendix:Extreme Scenarios} contains the analysis of CPL in \textbf{extreme scenarios}.

\begin{lemma}\label{lemma:sigma_greater_than_ratio}
Let $0 \leq A, B, g' \leq 1$, and suppose $e^\varepsilon > 1$. Then,
   \begin{equation*}
       \frac{1+(A+\sigma g')\lambda}{1+(B+g')\lambda} -\frac{1+A\lambda}{1+B\lambda} \geq 0 \iff \sigma \geq \frac{1+A\lambda}{1+B\lambda}.
   \end{equation*}
   Here, $\lambda = (e^\varepsilon-1)$.
\end{lemma}

The proof of Lemma~\ref{lemma:sigma_greater_than_ratio} is available in Appendix~\ref{proof:lemma:sigma_greater_than_ratio}.

\begin{theorem}\label{Theorem:optimality_of_algo}
    Algorithm~\ref{algorithm_2_part_2} always converges to the optimal $\overline{H}^*$ for selected $G,G'$.
\end{theorem}

The proof of Theorem~\ref{Theorem:optimality_of_algo} is available in Appendix~\ref{proof:Theorem:optimality_of_algo}.



\subsection{Theoretical Explanation for the Observations of Statistical Analysis}\label{section:Theoretical Explanation for the Observations of Statistical Analysis}

In this section, we theoretically explain the statistical observations (i.e., OB1, OB2, and OB3) in Section~\ref{section:limitations_of_the_statistical_evaluation_method}.

\subsubsection{For \textbf{OB1}} Since $P(\Hat{X}|X_k)$ is not necessarily equal to $P(X_k|\Hat{X})$, the resulting CPL could be asymmetric according to the Algorithm~\ref{algorithm_2_part_2}. 

\subsubsection{For \textbf{OB2}}According to Corollary~\ref{corollary:max_privacy_leakage_privacy_budget}, CPL saturates as $\varepsilon$ increases if all the probabilities are non-zero in the joint probability distribution of two attributes. 



\subsubsection{For \textbf{OB3}} \label{section:theoretical_explaination_OB3}
This is because CPL of an attribute only depends on its two realisations of the alphabet (i.e., $G,G'$).
However, conventional metrics represent the average value for the whole joint probability distribution.

The next section explores the potential applications of theoretical findings.

\section{Applications of CPL Analysis}\label{section:Applications of CPL Analysis}

This section demonstrates the potential applications of CPL analysis to improve distributed systems. In Section~\ref{section:Benchmark for CPL Analysis Methods} and Section~\ref{section:Utility vs CPL Benchmark for LDP}, we propose two novel benchmarks for CPL analysis algorithms in LDP and utility vs CPL in LDP mechanisms, respectively. Finally, we demonstrate how publicly available knowledge can be used to calibrate the privacy budgets in Section~\ref{section:Privacy Budget Calibration}. 

\subsection{Benchmark for CPL Analysis Algorithms}\label{section:Benchmark for CPL Analysis Methods}

\begin{figure*}[h]
  \centering
  \includegraphics[trim={2.6cm 0.2cm 0cm 0.2cm},clip,width=0.95\linewidth]{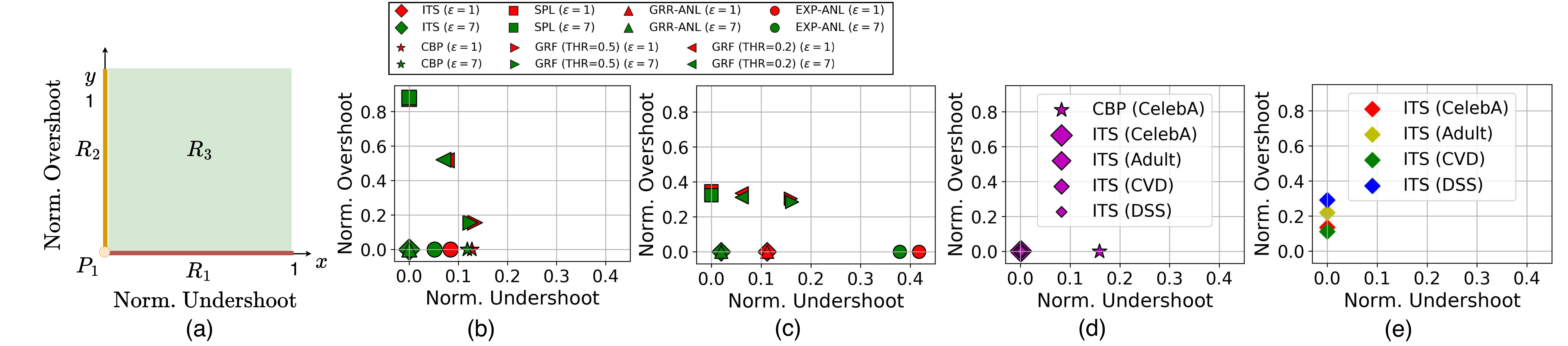}
  \caption{(a) The three regions ($R_1,R_2,R_3$) and the point ($P_1$) of the benchmark for CPL analysis algorithms. (b) and (c) Results of the benchmark for CPL analysis algorithms (ref. is $\varepsilon$-LDP using Algorithm~\ref{algorithm_2_part_2}) for the CelebA dataset and DSS dataset, respectively. In (b) and (c), the shapes of the markers correspond to different privacy mechanisms, and the colours correspond to different privacy budgets. (d) Results of the benchmark for CPL analysis algorithms (ref. is GRR using Algorithm~\ref{algorithm_1}) for all the datasets. (e) Results of the benchmark for CPL analysis algorithms (ref. is statistical OUE) for all the datasets.}
  \vspace{-1.5em}
\label{fig:benchmark_1}
\end{figure*}

Several algorithms have been proposed in the literature to quantify privacy leakage in distributed systems, including:

\begin{itemize}
\item \textbf{SPL-ANL} \cite{SPL_RS_RS_PLUS_FAKE}: A naive method estimating privacy leakage by \emph{assuming perfect correlation} among attributes, analogous to simple privacy budget splitting (SPL). If $n$ attributes are perturbed with privacy budgets $\varepsilon_1, \dots, \varepsilon_n$, the total estimated leakage is $\sum_{i \in [n]} \varepsilon_i$.

\item \textbf{GRF}\cite{LDP-PCC-dependency-graph-IoT}: Constructs a dependency graph based on conventional correlation metrics. It computes Pearson correlation coefficients (PCC) between attribute pairs and applies a threshold (THR) for connectivity. Attributes with absolute values of PCC below THR are considered uncorrelated. Then, SPL-ANL is applied to each connected sub-graph, with the maximum leakage across sub-graphs used as the estimate. Our experiments use THR = 0.2 and 0.4, marking the lower and upper bounds of weak correlation~\cite{Wikipedia_Correlation_Coefficient_thresholds}.

\item \textbf{ITS} \cite{data_correlation_location_similar_LDP}: Computes CPL assuming a two-level perturbation mechanism (e.g., GRR~\cite{LDP_Frequent_Itemset_Mining}, OUE~\cite{ldp_Frequency_Estimation}).

\item \textbf{CBP} \cite{secon_Collecting_High-Dimensional_Correlation_LDP}: Calibrates privacy budgets specifically for GRR mechanisms within a relaxed LDP framework (UDLDP). CBP is efficient primarily for binary attributes due to computational complexity scaling exponentially with alphabet size.
\end{itemize}

Despite these advancements, a critical question remains: \emph{``Do these algorithms accurately quantify CPL in distributed systems?''}.
To address this question, we propose a benchmark framework to evaluate CPL analysis algorithms. This benchmark evaluates algorithms under two criteria: (i) generic LDP mechanisms using Algorithm~\ref{algorithm_2_part_2}, and (ii) mechanism-specific benchmarks using Algorithm~\ref{algorithm_1}.

This benchmark provides a \textbf{standardised framework} for evaluating the \emph{optimality}, \emph{limitations}, and \emph{comparative performance} of CPL quantification algorithms across diverse real and synthetic datasets, facilitating rigorous assessment of future CPL analysis methods.

\subsubsection{Definition}

Let $q$ denote the total number of pairwise attribute combinations. For each combination $i$, let $l_i^*$ represent the reference CPL and $l_i$ denote the CPL estimated by a given algorithm. The algorithm's estimation performance is quantified using normalised overshoot and undershoot as follows:
\begin{equation}
\begin{split}
\text{Normalised Undershoot} &= \frac{\sum_{i \in [q], (l_i^*-l_i)\geq 0}(l_i^*-l_i)}{\varepsilon q},\\
\text{Normalised Overshoot} &= \frac{\sum_{i \in [q], (l_i^*-l_i)< 0}(l_i - l_i^*)}{\varepsilon q}.
\end{split}
\end{equation}
Here, $\varepsilon$ is the perturbation privacy budget of each attribute. The reference CPL $l_i^*$ can be computed using Algorithm~\ref{algorithm_1} or the statistical method (for specific LDP mechanisms, e.g., GRR, EXP) or Algorithm~\ref{algorithm_2_part_2} (for generic mechanisms). 
Further, the values are normalised by $\varepsilon$ and $q$, enabling comparisons across different privacy budgets and datasets. 

We represent \emph{Undershoot} in the $x-axis$ and \emph{Overshoot} in the $y-axis$. Each algorithm configuration corresponds to a unique point in this normalised space.
Then, the benchmark graph is partitioned into distinct evaluation regions (See Fig.~\ref{fig:benchmark_1}(a)):

\begin{enumerate}
\item [$\bullet$] \textbf{P1 (Origin, $x=0$, $y=0$):} Indicates optimal CPL estimation, matching exactly the reference CPL for the given LDP mechanism.

\item [$\bullet$] \textbf{R1 ($y=0,\ x>0$):} Represents consistent underestimation of CPL, signifying potential privacy risks due to underestimated leakage.

\item [$\bullet$] \textbf{R2 ($x=0,\ y>0$):} Reflects consistent overestimation, generally safer but potentially resulting in unnecessary utility loss.

\item [$\bullet$] \textbf{R3 ($x>0,\ y>0$):} Denotes mixed estimation, indicating simultaneous presence of overestimation and underestimation across attribute combinations.
\end{enumerate}

Additionally, since the $x$ and $y$ coordinates represent the normalised average undershoot and overshoot with respect to the reference CPL, the distance from the origin is proportional to the deviation from the reference.

\subsubsection{Results}
\textbf{(i) Comparison with Generic $\varepsilon$-LDP Mechanisms.} We benchmark SPL-ANL, GRF, ITS, CBP, GRR-ANL, and EXP-ANL using CelebA and DSS datasets (see Fig.~\ref{fig:benchmark_1}(b,c)
).
Here, \emph{GRR-ANL} and \emph{EXP-ANL} denote the estimated CPL by Algorithm~\ref{algorithm_1} (i.e., actual CPL) for GRR and EXP mechanisms. 

Then, the naive method \emph{SPL-ANL} consistently lies in R2, confirming it overestimates CPL, because SPL-ANL assumes all the attributes are perfectly correlated with each other.
\emph{GRF} (THR=0.2, 0.4) positions in R3, which implies some of the estimated CPL are underestimated and some are overestimated. This observation confirms the limitations in conventional correlation metrics for accurate CPL characterisation (recall that GRF uses PCC and a predefined threshold to identify and filter the correlated attributes). \emph{CBP} lies in R1, which implies that CBP underestimates CPL than generic $\varepsilon$-LDP (i.e., reference). Therefore, CBP is \textbf{only} valid for a set of LDP mechanisms.
Interestingly, \emph{ITS} estimates the optimal (P1) CPL for the CelebA dataset (Fig.~\ref{fig:benchmark_1}(b)), which implies that the estimated CPL of ITS is valid for \textbf{any} $\varepsilon$-LDP mechanisms for the CelebA dataset. This insight was not previously identified as authors focused on only two-level perturbation mechanisms for ITS~\cite{data_correlation_location_similar_LDP}. However, ITS underestimates (R1) for the DSS dataset (Fig.~\ref{fig:benchmark_1}(c)), which implies that ITS is \textbf{only} valid for a set of LDP mechanisms for the DSS dataset. Here, note that CelebA is a binary dataset and DSS is a non-binary dataset (consists of attributes that have a range of alphabet sizes). Interestingly, \emph{GRR-ANL} places in the same positions as ITS in the results of both CelebA and DSS datasets as depicted in Fig.~\ref{fig:benchmark_1}(b,c). \emph{EXP-ANL} lies in R1, which aligns with the low CPL of EXP mechanisms observed in the statistical results in Fig.~\ref{fig:motivation_results}(a) in Section~\ref{section:CPL between Attributes in Adult Dataset with Different LDP Mechanisms}.

Next, we benchmark ITS and CBP with respect to the actual CPL of the GRR mechanism (using Algorithm~\ref{algorithm_1}) using all the datasets (note that CBP is only applicable for binary dataset CelebA because of the exponential time complexity with alphabet size) to check whether these methods correctly estimate the CPL for GRR, as claimed by the authors.

\textbf{(ii) Comparison with Specific $\varepsilon$-LDP Baseline.} 

Fig.~\ref{fig:benchmark_1}(d) depicts the benchmark of ITS and CBP mechanisms with respect to the actual CPL of GRR (using Algorithm~\ref{algorithm_1}).
According to the results, ITS aligns optimally (P1) for all datasets, accurately reflecting the actual CPL, which verifies that the ITS correctly estimates the CPL of GRR mechanisms as claimed by the authors.
However, we notice that CBP notably underestimates (R1), which contradicts the claim of the authors. Therefore, we have delved into the CBP algorithm further to identify the reason for the underestimated CPL. We identify that, in CBP algorithm development (Definition 3.1 in~\cite{secon_Collecting_High-Dimensional_Correlation_LDP}), authors have assumed that the privacy leakage between attributes is symmetric, which may lead to underestimated privacy leakage in the CBP method.

Additionally, we evaluate ITS against OUE (ref. is statistically estimated CPL) across all datasets as depicted in Fig.~\ref{fig:benchmark_1}(e). Here, ITS consistently overestimates (R2). Therefore, we can conclude that ITS is optimised for GRR mechanisms and overestimates some two-level perturbation mechanisms such as OUE. However, ITS can be used to estimate the CPL of OUE mechanisms, as overestimation does not lead to privacy vulnerabilities in the distributed system. 

In conclusion, ITS estimates the CPL of GRR and OUE (the estimated CPL for OUE is not optimal). CBP underestimates the CPL of GRR, which could lead to privacy vulnerabilities. SPL-ANL significantly overestimates the CPL. GRF indicates a combination of overestimates and underestimates of CPL, which verifies the inaccurate CPL characterisation of the PCC metric.

\subsection{Utility vs TCPL Benchmark for LDP Mechanisms}\label{section:Utility vs CPL Benchmark for LDP}


LDP mechanisms may vary in terms of utility and CPL. 
In this benchmark, we evaluate the performance of LDP mechanisms in terms of data utility error vs CPL for the ground truth of the data. Identifying mechanisms that achieve the desired utility with minimal CPL allows better utilisation of privacy budgets in data collections in distributed systems such as IoT devices.

\subsubsection{Definition}

We use normalised total correlation-induced privacy leakage (normalised TCPL, see Section~\ref{section:Total Correlation-Induced Privacy Leakage} for the definition of TCPL) to measure the total privacy leakage of the mechanisms.

\noindent\textbf{Normalised Total Correlation-Induced Privacy Leakage:}

The normalised TCPL for an LDP mechanism is given by,
\begin{equation}
\text{Norm. TCPL} = \frac{\operatorname{TCPL}'}{\operatorname{TCPL}^*} = \frac{\sum_{x_k \in X} \sum_{x \in X\setminus \{x_k\}} L'_{x \rightarrow x_k}}{\sum_{x_k \in X} \sum_{x \in X\setminus \{x_k\}} L^*_{x \rightarrow x_k}}.
\end{equation}
Here, $\operatorname{TCPL}'$ denotes the TCPL of the LDP mechanism for the selected dataset (can be computed with Algorithm~\ref{algorithm_1} or statistical method), and $\operatorname{TCPL}^*$ denotes the TCPL given by Algorithm~\ref{algorithm_2_part_2} for the selected dataset.

Moreover, we have used two commonly adopted utility error metrics in LDP data collection applications to measure the utility of the mechanisms:

\noindent\textbf{Frequency Estimation Error:}

Normalised mean square error (NMSE) of frequency estimation in $X$ attributes is given by,
\begin{equation}
    \operatorname{Freq. Est. Error} = \frac{\sum_{x \in X} \bigl(\tilde{f}_x - f_x\bigr)^2}{ \sum_{x \in X}f_x ^2}.
\end{equation}
Here, $f_x$ and $\tilde{f}_x$ are actual and estimated frequency of $x^{th}$ attribute.

\noindent\textbf{Normalised 0-1 Error:}

Normalised 0-1 error between the original inputs $X$ and the  perturbed outputs $Y$ is given by,
\begin{equation}
    \text{Normalised 0-1 Error} = \frac{\sum_{i = 1}^n f(x_i,y_i)}{n},
\end{equation}
where $x_i \in X,\ y_i \in Y$ for all $i \in [n], \ |X|=|Y|=n$. Here, $f(x_i,y_i) = 0$ if $x_i = y_i$. Otherwise, $f(x_i,y_i) = 1$.

\subsubsection{Results}

\begin{figure}[h]
  \centering
  \includegraphics[trim={0.3cm 0.2cm 0.3cm 0.2cm},clip,width=0.9\linewidth]{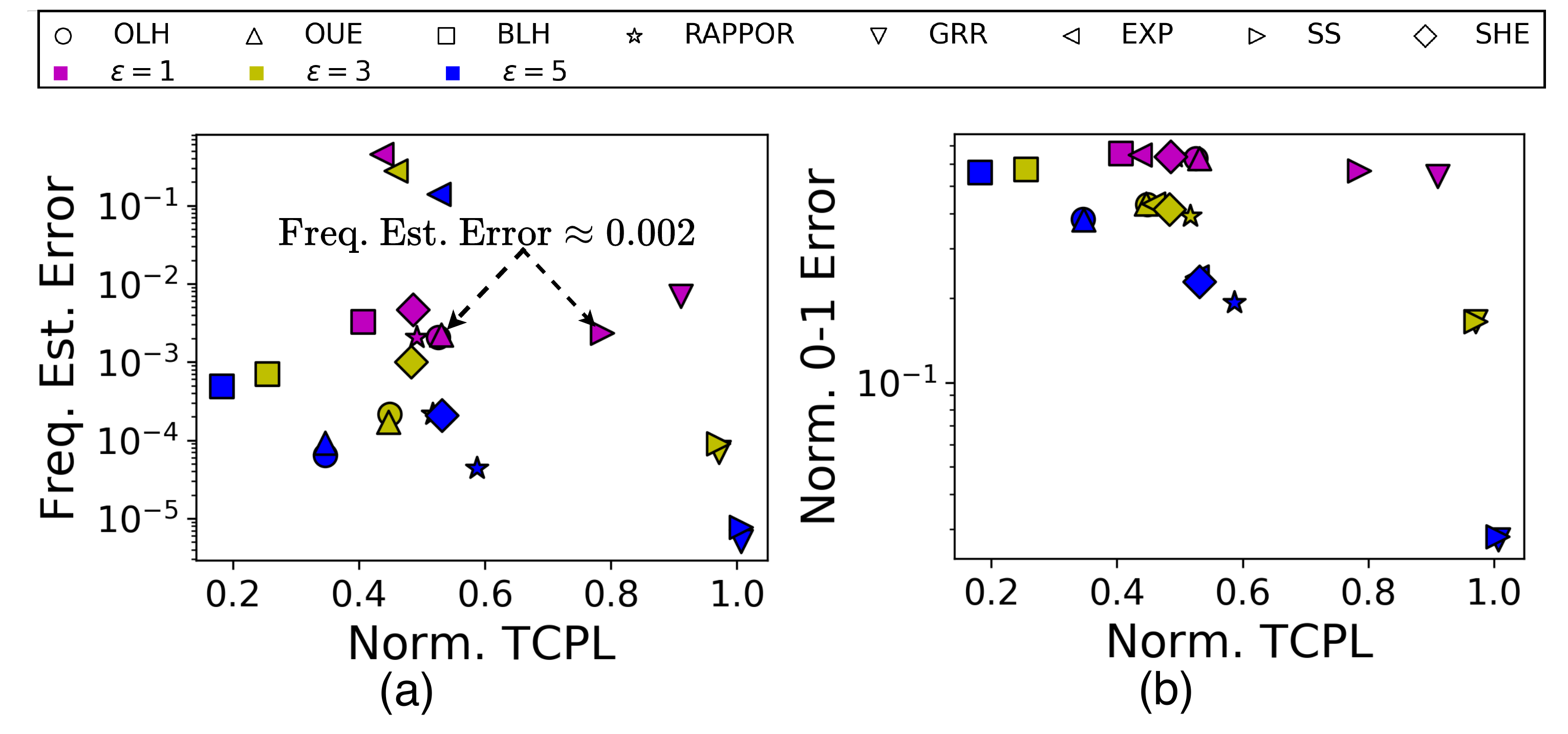}
  \vspace{-0.7em}
  \caption{Utility vs CPL benchmark of LDP mechanisms for the Adult dataset. (a) Benchmark for frequency estimation error vs TCPL.
  (b) Benchmark for normalised 0-1 error vs TCPL.
  }
  \vspace{-1.5em}
  \label{fig:benchmark_2}
\end{figure}

We benchmark eight widely used LDP mechanisms: GRR\cite{LDP_Frequent_Itemset_Mining}, RAPPOR\cite{rappor_original}, EXP~\cite{Algorithmic_Foundations_dwork}, OUE\cite{ldp_Frequency_Estimation}, BLH~\cite{ldp_Frequency_Estimation}, OLH~\cite{ldp_Frequency_Estimation}, Summation with histogram encoding (SHE)~\cite{dwork_she}, and Subset selection (SS)~\cite{subsetselection_min} (See Appendix~\ref{section:appendix:ldp_mechanisms} for more detail about the mechanisms). These mechanisms have been adopted in a broad range of applications, including frequency estimation~\cite{ldp_Frequency_Estimation},
and heavy hitter detection~\cite{Heavy_Hitter_Estimation_LDP}. 

First, we evaluate frequency estimation error vs normalised TCPL of all eight mechanisms for the Adult dataset under three different privacy budgets $\varepsilon=1,3,5$ as depicted in Fig.~\ref{fig:benchmark_2}(a). We observe that the LDP mechanisms have different normalised TCPL for different utility levels at the same privacy budget.
Interestingly, we observe that the OLH mechanism achieves a frequency estimation error of approximately $0.002$, while incurring only $60\%$ of the normalised TCPL compared to the SS mechanism when $\varepsilon=1$.
Moreover, these results are consistent with other datasets (CelebA, CVD, and DSS). Normalised 0-1 error vs normalised TCPL also shows similar behaviour as depicted in Fig.~\ref{fig:benchmark_2}(b).


\subsection{Correlation-Aware Privacy Budget Calibration}\label{section:Privacy Budget Calibration}

Privacy budget calibration is a process of optimally adjusting individual privacy budgets to balance the trade-off between privacy and utility in multi-dimensional data. 
For example, suppose given a maximum total privacy budget $\overline{\varepsilon}$, our goal is to determine the largest possible per-attribute budget $\tilde{\varepsilon}$ such that the total privacy leakage (TPL) for each attribute remains within the predefined bound. 
Privacy budget calibration can be done either considering or without considering the correlation information.
The naive approach that does not use correlation information equally distributes the total privacy budget (i.e., $\overline{\varepsilon}$) across attributes (denoted as SPL).
We propose a novel algorithm (See Appendix~\ref{section:appendix:privacy_budget_calibration}) to compute the optimal privacy budget $\varepsilon^*$ to perturb each attribute using CPL analysis. Here, we assume that all attributes are assigned equal privacy budgets.

Next, we evaluate the correlation-aware privacy budget calibration method with the naive method (SPL) using the CelebA dataset.
Here, we have selected 21 attributes from the CelebA dataset, which are common with another public dataset, Labelled Faces in the Wild (LFW)~\cite{lfw_attributes} (we use the LFW dataset as a publicly available information to learn correlations in the CelebA dataset).

First, we compute the calibrated privacy budget across different $\overline{\varepsilon}$ assuming we have the ground truth probability information of the CelebA as shown in the red dashed line in Fig.~\ref{fig:celeba_lfw_calibration_results}(a). Compared to the naive method (SPL), the correlation-aware calibrated privacy budget is approximately five times higher, allowing the data to be perturbed with significantly less noise for the same privacy requirements. Next, we evaluate the utility error of perturbed data at each privacy requirement (i.e., $\overline{\varepsilon}$) as depicted in Fig.~\ref{fig:celeba_lfw_calibration_results}(b). Here, we adopted empirical utility calculation by perturbing the CelebA dataset using GRR (the simplest mechanism) and computed 0-1 error between the data. As expected, the correlation-aware calibrated privacy budget has much less utility error compared to the SPL method.

Next, we use the LFW dataset to estimate CPL and apply our calibration algorithm accordingly. As shown in Fig.~\ref{fig:celeba_lfw_calibration_results}(a), the calibrated privacy budget using public knowledge (i.e., LFW dataset) is close to the calibrated privacy budget using ground truth. \emph{This experimental result demonstrates the effectiveness of leveraging public correlation knowledge for privacy budget calibration}.

Beyond the above applications, our statistical and theoretical findings (Sections~\ref{section:motivation} and~\ref{section:methodology}) open new directions for designing LDP mechanisms and tools that balance privacy and utility by leveraging CPL analysis with prior knowledge.

\begin{figure}[bth]
  \centering
  \vspace{-1em}
  \includegraphics[trim={1.2cm 0.5cm 1cm 0.2cm},clip,width=0.9\linewidth]{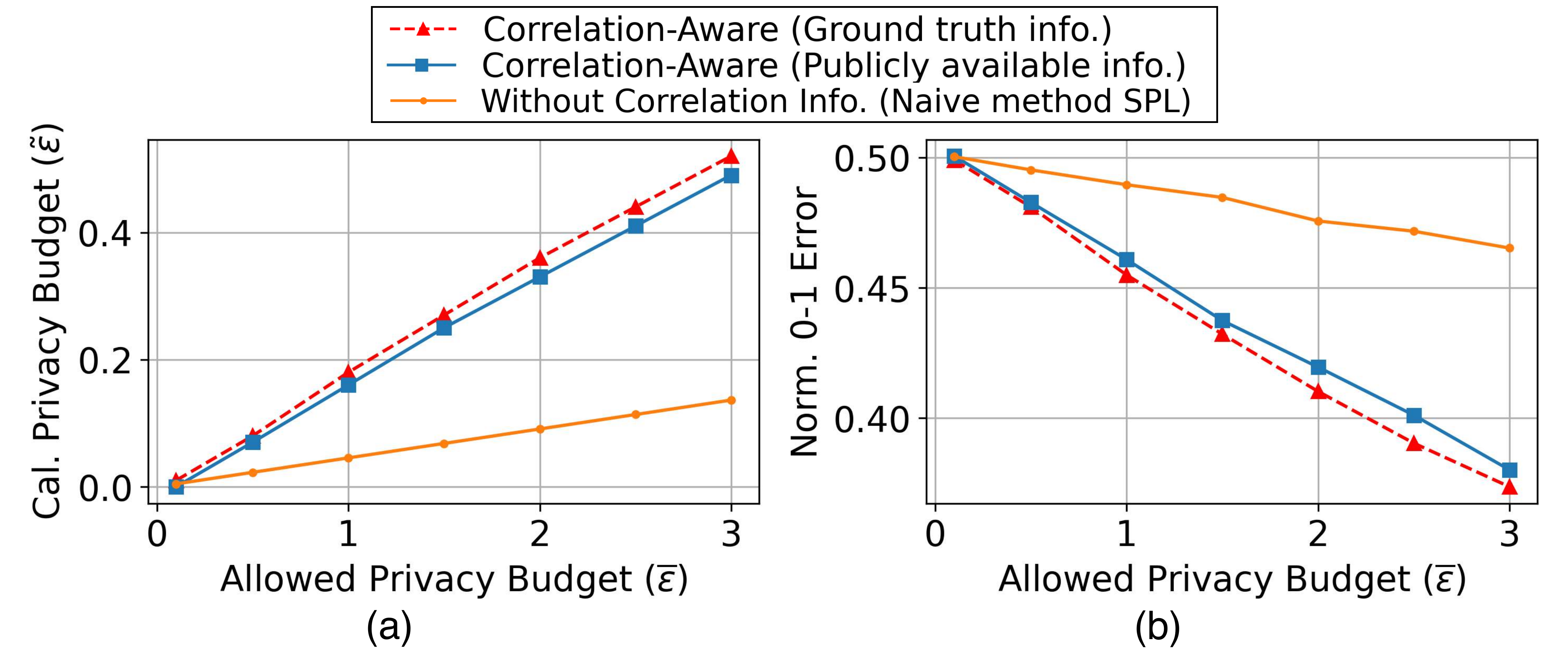}
  \caption{(a) Calibrated privacy budget of single attribute ($\tilde{\varepsilon}$) vs allowed privacy budget for whole dataset ($\overline{\varepsilon}$) for different algorithms. (b) Normalised 0-1 error vs allowed privacy budget for whole dataset ($\overline{\varepsilon}$) for different algorithms.}
  \vspace{-1.5em}
  \label{fig:celeba_lfw_calibration_results}
\end{figure}




\section{Related Work}\label{section:related_works}

First, we review existing privacy mechanisms associated with Differential Privacy (DP) and Local Differential Privacy (LDP). Next, we examine prior knowledge and its associated privacy risks. Finally, we review DP and LDP mechanisms for correlated data. 

\subsection{DP and LDP}

DP and LDP have become de facto tools for preserving privacy due to the limitations in traditional privacy algorithms such as k-anonymity~\cite{k-anonymity}, l-diversity~\cite{l-diversity} etc.~\cite{Algorithmic_Foundations_dwork, LDP_survey_composition_theorem}. As we mentioned before, researchers adopt DP and LDP mechanisms based on the requirements. LPD has become a popular method when the data collector is untrustworthy~\cite{LDP_survey_composition_theorem, apple2017privacy, microsoft_ldp, PCKV_key_value, LDP_Frequent_Itemset_Mining}. 
A recent study~\cite{PETS_statistical_audit} has utilised the statistical analysis to audit the LDP mechanisms to verify the privacy guarantee of LDP mechanisms using statistical analysis and synthetic data. However, this work does not explore the effect of correlation and privacy leakage caused by the correlations in their analysis. Both LDP and DP are initially defined by assuming that the attributes of data have no relationship with each other. However, most real-world data is inherently correlated to some extent, and most of these relationships could be freely available for adversaries~\cite{impact_of_prior_knowledge}. As a result, the privacy guarantee of (L)DP may not be accurate anymore in most practical scenarios~\cite{no_free_lunch, correlated_DP}.

\subsection{Prior Knowledge}

Prior knowledge could have different forms depending on the data type: distributional details~\cite{correlated_DP, impact_of_prior_knowledge}, information from previous timestamps in time-series data~\cite{ data_correlation_location_similar_LDP}, current location based on last location~\cite{data_correlation_location_similar_LDP} etc. Literature shows that this knowledge can be gained in various methods: from universal norms~\cite{impact_of_prior_knowledge}, from public datasets~\cite{adult_dataset, celebA_attributes}, learning while collecting data~\cite{AAA} etc. This leads to increased privacy risks and inaccurate privacy guarantees for mechanisms that do not explicitly consider prior knowledge~\cite{impact_of_prior_knowledge, correlated_DP}.

In this light, there has been some work improving the privacy guarantee leveraging prior knowledge~\cite{secon_Collecting_High-Dimensional_Correlation_LDP, AAA, correlated_DP}. Next, we summarise how studies have incorporated prior knowledge to enhance the privacy guarantee in (L)DP.

\subsection{(L)DP Mechanisms for Correlated Data}

\subsubsection{DP Related Work}

Considerable work has been done on correlated data under the central setting of differential privacy~\cite{Algorithmic_Foundations_dwork, correlated_DP}. Group DP is a simple method to adjust the global sensitivity based on the number of correlated records~\cite{Algorithmic_Foundations_dwork}. Instead of this global sensitivity, \cite{correlated_DP} formulates the sensitivity parameter using a correlation degree matrix. Another study~\cite{correlated_data_dp_conventional_metrics_MI} uses distance correlation to measure the correlation. In~\cite{dynamic_sensitivity_change_dp}, a periodic sensitivity calculation algorithm is proposed to address the periodic correlation changes in temporal data using autocorrelation. However, in these studies, correlation is measured using conventional metrics such as PCC, and the authors assume that the privacy leakage impact between user records is symmetric without a proper theoretical explanation for the correctness of these metrics. Additionally, recent works focus on analysing the privacy leakage in correlated temporal and location data~\cite{May_the_privacy_location_correlation_DP, Quantifying_DP}.
Interestingly, we notice that study~\cite{Quantifying_DP} also have obtained a \emph{mathematically} similar form of the optimisation problem to (\ref{eqn:H}) (when $\delta=0$), although they analyse the temporal data in the central setting of DP.
 However, these studies cannot apply to $(\varepsilon, \delta)$-LDP mechanisms as these works analyse the privacy leakage under the central setting of pure DP ($\varepsilon$-DP). 

Additionally, several studies have proposed frameworks to address the data correlation by improving the DP definition. Pufferfish~\cite{pufferfish}
and Bayesian DP~\cite{bayesian_differential_privacy} are notable such frameworks. Since these frameworks are defined for the central setting, they are not directly applicable to LDP and the proposed algorithms in this paper.

\subsubsection{LDP Related Work}


Several studies have proposed novel mechanisms to perturb correlated data~\cite{PCKV_key_value, secon_Collecting_High-Dimensional_Correlation_LDP}. In~\cite{PCKV_key_value}, two interdependent mechanisms have collected correlated key-value pairs. Studies~\cite{LDP-PCC-dependency-graph-IoT} and~\cite{secon_Collecting_High-Dimensional_Correlation_LDP} adopt the partitioning the privacy budget method, also known as privacy budget splitting (SPL),  to incorporate correlated data by measuring the strength of the correlation via MI and PCC. However, our study proves that these conventional metrics do not characterise LDP privacy leakage correctly. As a result, there could be a privacy risk as illustrated in the example in Section~\ref{section:theoretical_explaination_OB3}. Moreover, these studies focused on designing mechanisms rather than analysing privacy leakage when data is correlated.

Compared to DP, only a limited number of studies have been done on privacy leakage analysis on correlated data in LDP. A recent study~\cite{data_correlation_location_similar_LDP} has analysed privacy leakage related to spatial-temporal data. The authors have derived a general formulation (Equation 8 in Definition 5~\cite{data_correlation_location_similar_LDP}) similar to ours (\ref{eqn:expanded_additional_privacy_leakage}), which is used to compute temporal correlation-related privacy leakage. Therefore, this work~\cite{data_correlation_location_similar_LDP} verifies the formulation of the work. However, ~\cite{data_correlation_location_similar_LDP} solves this general formation only for a specific family of $\varepsilon$-LDP mechanisms which have two levels of perturbation probabilities (i.e. GRR~\cite{LDP_survey_composition_theorem}, Optimal unary encoding (OUE)~\cite{LDP_Frequent_Itemset_Mining}, etc.). In contrast, we analyse for any general $(\varepsilon,\delta)$-LDP mechanism, and we propose Algorithm~\ref{algorithm_1} (using transition probabilities of the LDP mechanism) and Algorithm~\ref{algorithm_2_part_2} (using $\varepsilon$ and $\delta$ when unknown transition probabilities of the LDP mechanism) to compute the correlation-induced privacy leakage (CPL).

In summary, no comprehensive work has been done to analyse correlation-related privacy leakage in terms of practical and theoretical perspectives in LDP. Limited available studies only consider specific categories of LDP that cannot be generalised for complex LDP mechanisms (e.g. Exponential~\cite{Algorithmic_Foundations_dwork}, RAPPOR~\cite{rappor_original}, etc.) and approximated LDP (i.e., ($\varepsilon$, $\delta$)-LDP).

\section{Conclusion}\label{section:conclusion}

Addressing correlation-induced privacy leakage (CPL) in distributed systems under Local Differential Privacy (LDP) remains challenging. This study combines statistical analysis on four real-world datasets (Section~\ref{section:motivation}) with theoretical investigation, proposing two novel algorithms (Algorithm~\ref{algorithm_1} and Algorithm~\ref{algorithm_2_part_2}) to quantify CPL in general ($\varepsilon,\delta$)-LDP scenarios. Our results indicate that severe CPL is rare in practice, suggesting strategic privacy budget allocation can enhance utility.
We demonstrate that conventional metrics like mutual information (MI) and Pearson correlation coefficient (PCC) inadequately capture CPL, potentially increasing privacy risks. Next, we propose two benchmarks for evaluating distributed LDP analysis methods (Section~\ref{section:Benchmark for CPL Analysis Methods}) and mechanisms (Section~\ref{section:Utility vs CPL Benchmark for LDP}), highlighting shortcomings and improvement opportunities.

Finally, we illustrate CPL analysis application for privacy budget calibration in distributed data collection using public correlation data (Section~\ref{section:Privacy Budget Calibration}).

\subsection{Limitations and Future Work}
The proposed algorithms assume that the characteristics of the data distribution are known. Although many studies have shown that the underlying data distribution could be learned in different ways, e.g. publicly available datasets, collecting a portion of population data to learn distribution, etc. ~\cite{impact_of_prior_knowledge, AAA}, it is possible that the assumed distribution deviates from the actual data distribution. This could lead to suboptimal CPL. Therefore, the consequences of deviation in the assumed data distribution will be investigated further in future work. 

\section{Online Resources}

Artifacts will be made available soon.




\bibliographystyle{IEEEtran}
%

\bibliography{paper/refs}

\begin{thebibliography}{10}
\providecommand{\url}[1]{#1}
\csname url@samestyle\endcsname
\providecommand{\newblock}{\relax}
\providecommand{\bibinfo}[2]{#2}
\providecommand{\BIBentrySTDinterwordspacing}{\spaceskip=0pt\relax}
\providecommand{\BIBentryALTinterwordstretchfactor}{4}
\providecommand{\BIBentryALTinterwordspacing}{\spaceskip=\fontdimen2\font plus
\BIBentryALTinterwordstretchfactor\fontdimen3\font minus \fontdimen4\font\relax}
\providecommand{\BIBforeignlanguage}[2]{{%
\expandafter\ifx\csname l@#1\endcsname\relax
\typeout{** WARNING: IEEEtran.bst: No hyphenation pattern has been}%
\typeout{** loaded for the language `#1'. Using the pattern for}%
\typeout{** the default language instead.}%
\else
\language=\csname l@#1\endcsname
\fi
#2}}
\providecommand{\BIBdecl}{\relax}
\BIBdecl

\bibitem{LDP-PCC-dependency-graph-IoT}
J.~Wei, J.~Li, Y.~Lin, and J.~Zhang, ``Ldp-based social content protection for trending topic recommendation,'' \emph{IEEE Internet of Things Journal}, vol.~8, no.~6, pp. 4353--4372, 2021.

\bibitem{apple2017privacy}
A.~Inc., ``Learning with privacy at scale,'' \url{https://machinelearning.apple. com/research/learning-with-privacy-at-scale}, 2017, accessed: 2025.

\bibitem{PCKV_key_value}
X.~Gu, M.~Li, Y.~Cheng, L.~Xiong, and Y.~Cao, ``{PCKV}: Locally differentially private correlated {Key-Value} data collection with optimized utility,'' in \emph{USENIX Security 20}, Aug. 2020.

\bibitem{rappor_original}
U.~Erlingsson, V.~Pihur, and A.~Korolova, ``Rappor: Randomized aggregatable privacy-preserving ordinal response,'' in \emph{Proceedings of the 2014 ACM SIGSAC CCS}, ser. CCS '14.\hskip 1em plus 0.5em minus 0.4em\relax New York, NY, USA: ACM, 2014.

\bibitem{LDP_survey_composition_theorem}
T.~Wang, X.~Zhang, J.~Feng, and X.~Yang, ``A comprehensive survey on local differential privacy toward data statistics and analysis,'' \emph{Sensors}, vol.~20, no.~24, 2020.

\bibitem{LDP_duchi}
J.~C. Duchi, M.~I. Jordan, and M.~J. Wainwright, ``Local privacy and statistical minimax rates,'' in \emph{2013 51st Annual Allerton Conference on Communication, Control, and Computing}, 2013, pp. 1592--1592.

\bibitem{microsoft_ldp}
B.~Ding, J.~Kulkarni, and S.~Yekhanin, ``Collecting telemetry data privately,'' ser. NIPS'17.\hskip 1em plus 0.5em minus 0.4em\relax Red Hook, NY, USA: Curran Associates Inc., 2017, p. 3574–3583.

\bibitem{Algorithmic_Foundations_dwork}
C.~Dwork and A.~Roth, ``The algorithmic foundations of differential privacy,'' \emph{Foundations and Trends® in Theoretical Computer Science}, vol.~9, no. 3–4, pp. 211--407, 2014.

\bibitem{no_free_lunch}
D.~Kifer and A.~Machanavajjhala, ``No free lunch in data privacy,'' in \emph{Proceedings of the 2011 ACM SIGMOD International Conference on Management of Data}, ser. SIGMOD '11.\hskip 1em plus 0.5em minus 0.4em\relax New York, NY, USA: ACM, 2011, p. 193–204.

\bibitem{SPL_RS_RS_PLUS_FAKE}
\BIBentryALTinterwordspacing
H.~H. Arcolezi, S.~Gambs, J.-F. Couchot, and C.~Palamidessi, ``On the risks of collecting multidimensional data under local differential privacy,'' \emph{Proc. VLDB Endow.}, vol.~16, no.~5, p. 1126–1139, Jan. 2023. [Online]. Available: \url{https://doi.org/10.14778/3579075.3579086}
\BIBentrySTDinterwordspacing

\bibitem{RS_plus_FD}
H.~H. Arcolezi, J.-F. Couchot, B.~Al~Bouna, and X.~Xiao, ``Random sampling plus fake data: Multidimensional frequency estimates with local differential privacy,'' in \emph{Proceedings of the 30th ACM CIKM}, ser. CIKM '21.\hskip 1em plus 0.5em minus 0.4em\relax New York, USA: ACM, 2021, p. 47–57.

\bibitem{impact_of_prior_knowledge}
Y.~Li, X.~Ren, S.~Yang, and X.~Yang, ``Impact of prior knowledge and data correlation on privacy leakage: A unified analysis,'' \emph{IEEE Transactions on Information Forensics and Security}, vol.~14, no.~9, pp. 2342--2357, 2019.

\bibitem{adult_dataset}
B.~Becker and R.~Kohavi, ``{Adult},'' UCI Machine Learning Repository, 1996, {DOI}: https://doi.org/10.24432/C5XW20.

\bibitem{celebA_attributes}
Z.~Liu, P.~Luo, X.~Wang, and X.~Tang, ``Deep learning face attributes in the wild,'' in \emph{ICCV}, December 2015.

\bibitem{AAA}
F.~Wei, E.~Bao, X.~Xiao, Y.~Yang, and B.~Ding, ``Aaa: An adaptive mechanism for locally differentially private mean estimation,'' \emph{Proc. VLDB Endow.}, vol.~17, no.~8, p. 1843–1855, May 2024.

\bibitem{data_correlation_location_similar_LDP}
K.~M. Chong and A.~Malip, ``Local differential privacy for correlated location data release in its,'' \emph{Computer Networks}, vol. 255, p. 110830, 2024.

\bibitem{secon_Collecting_High-Dimensional_Correlation_LDP}
R.~Du, Q.~Ye, Y.~Fu, and H.~Hu, ``Collecting high-dimensional and correlation-constrained data with local differential privacy,'' in \emph{2021 18th Annual IEEE International Conference on Sensing, Communication, and Networking (SECON)}, 2021, pp. 1--9.

\bibitem{LDP_Frequent_Itemset_Mining}
T.~Wang, N.~Li, and S.~Jha, ``Locally differentially private frequent itemset mining,'' in \emph{2018 IEEE Symposium on Security and Privacy (SP)}, 2018, pp. 127--143.

\bibitem{ldp_Frequency_Estimation}
T.~Wang, J.~Blocki, N.~Li, and S.~Jha, ``Locally differentially private protocols for frequency estimation,'' in \emph{USENIX Security 17}.\hskip 1em plus 0.5em minus 0.4em\relax Vancouver, BC: USENIX Association, Aug. 2017, pp. 729--745.

\bibitem{correlated_DP}
T.~Zhu, P.~Xiong, G.~Li, and W.~Zhou, ``Correlated differential privacy: Hiding information in non-iid data set,'' \emph{IEEE Transactions on Information Forensics and Security}, vol.~10, pp. 229--242, 2015.

\bibitem{correlated_data_dp_conventional_metrics_MI}
S.~Biswas, A.~Fole, N.~Khare, and P.~Agrawal, ``Enhancing correlated big data privacy using differential privacy and machine learning,'' \emph{Journal of Big Data}, vol.~10, no.~1, p.~30, 03 2023, accessed: 2025.

\bibitem{wikipedia_pearson}
\BIBentryALTinterwordspacing
W.~contributors, ``Pearson correlation coefficient {Wikipedia}{,} the free encyclopedia,'' 2025. [Online]. Available: \url{https://en.wikipedia.org/wiki/Pearson\_correlation\_coefficient}
\BIBentrySTDinterwordspacing

\bibitem{moore2010basic}
D.~Moore, \emph{The Basic Practice of Statistics}.\hskip 1em plus 0.5em minus 0.4em\relax Freeman, 2010.

\bibitem{distinguishability_attacks_2}
H.~H. Arcolezi, S.~Gambs, J.-F. Couchot, and C.~Palamidessi, ``On the risks of collecting multidimensional data under local differential privacy,'' \emph{Proc. VLDB Endow.}, vol.~16, no.~5, p. 1126–1139, Jan. 2023.

\bibitem{PETS_statistical_audit}
H.~H. Arcolezi and S.~Gambs, ``Revealing the true cost of locally differentially private protocols: An auditing perspective,'' \emph{Proc. Priv. Enhancing Technol.}, vol. 2024, no.~4, pp. 123--141, 2024.

\bibitem{pureLDP2021}
G.~Cormode, S.~Maddock, and C.~Maple, ``pure‑ldp: Python implementations of state‑of‑the‑art local differential privacy algorithms,'' 2021.

\bibitem{Arcolezi2022-multi-freq-ldpy}
H.~H. Arcolezi, J.-F. Couchot, S.~Gambs, C.~Palamidessi, and M.~Zolfaghari, ``Multi‑freq‑ldpy: Multiple frequency estimation under local differential privacy in python,'' in \emph{Computer Security – ESORICS 2022}.\hskip 1em plus 0.5em minus 0.4em\relax Springer Nature Switzerland, 2022, pp. 770--775.

\bibitem{sulianova2021cardiovascular}
S.~Sulianova, ``Cardiovascular disease dataset,'' \url{https://www.kaggle.com/ datasets/sulianova/cardiovascular-disease-dataset}, 2021, accessed: 2025.

\bibitem{dss}
A.~U. R. I.~N. (AURIN), ``Dss 2016 dataset,'' \url{https://data.gov.au/dataset/ ds-aurin-f32b75e9-7735-4881-9884-d50b53d4fbe8}, accessed 2025.

\bibitem{covers_mutual_information}
\emph{Entropy, Relative Entropy and Mutual Information}.\hskip 1em plus 0.5em minus 0.4em\relax John Wiley and Sons, Ltd, ch.~2, pp. 12--49.

\bibitem{WikipediaSortingAlgorithm}
\BIBentryALTinterwordspacing
Wikipedia\_contributors, ``Sorting algorithm wikipedia, the free encyclopedia,'' accessed 2025. [Online]. Available: \url{\url{https://en.wikipedia.org/wiki/Sorting\_algorithm}}
\BIBentrySTDinterwordspacing

\bibitem{Wikipedia_Correlation_Coefficient_thresholds}
\BIBentryALTinterwordspacing
------, ``Correlation coefficient wikipedia, the free encyclopedia,'' 2025, accessed 2025. [Online]. Available: \url{https://en.wikipedia.org/wiki/Correlation\_coefficient}
\BIBentrySTDinterwordspacing

\bibitem{dwork_she}
C.~Dwork, F.~McSherry, K.~Nissim, and A.~Smith, ``Calibrating noise to sensitivity in private data analysis,'' in \emph{Proceedings of the Third Conference on TCC}.\hskip 1em plus 0.5em minus 0.4em\relax Springer-Verlag, 2006, p. 265–284.

\bibitem{subsetselection_min}
M.~Ye and A.~Barg, ``Optimal schemes for discrete distribution estimation under local differential privacy,'' in \emph{2017 IEEE International Symposium on Information Theory (ISIT)}, 2017, pp. 759--763.

\bibitem{Heavy_Hitter_Estimation_LDP}
Z.~Qin, Y.~Yang, T.~Yu, I.~M. Khalil, X.~Xiao, and K.~Ren, ``Heavy hitter estimation over set-valued data with local differential privacy,'' \emph{Proceedings of the 2016 ACM SIGSAC CCS}, 2016.

\bibitem{lfw_attributes}
N.~Kumar, A.~C. Berg, P.~N. Belhumeur, and S.~K. Nayar, ``Attribute and simile classifiers for face verification,'' in \emph{2009 IEEE 12th ICCV}, 2009, pp. 365--372.

\bibitem{k-anonymity}
L.~Sweeney, ``Achieving k-anonymity privacy protection using generalization and suppression,'' \emph{Int. J. Uncertain. Fuzziness Knowl.-Based Syst.}, vol.~10, no.~5, p. 571–588, Oct. 2002.

\bibitem{l-diversity}
A.~Machanavajjhala, J.~Gehrke, D.~Kifer, and M.~Venkitasubramaniam, ``L-diversity: privacy beyond k-anonymity,'' in \emph{22nd International Conference on Data Engineering (ICDE'06)}, 2006, pp. 24--24.

\bibitem{dynamic_sensitivity_change_dp}
\BIBentryALTinterwordspacing
J.~Zhao, S.~Liu, X.~Xiong, and Z.~Cai, ``Differentially private autocorrelation time-series data publishing based on sliding window,'' \emph{Secur. Commun. Networks}, vol. 2021, pp. 6\,665\,984:1--6\,665\,984:10, 2021. [Online]. Available: \url{https://api.semanticscholar.org/CorpusID:234365960}
\BIBentrySTDinterwordspacing

\bibitem{May_the_privacy_location_correlation_DP}
K.~M. Chong and A.~Malip, ``May the privacy be with us: Correlated differential privacy in location data for its,'' \emph{Comput. Netw.}, vol. 241, no.~C, Jun. 2024.

\bibitem{Quantifying_DP}
Y.~Cao, M.~Yoshikawa, Y.~Xiao, and L.~Xiong, ``Quantifying differential privacy under temporal correlations,'' in \emph{2017 IEEE 33rd International Conference on Data Engineering (ICDE)}, 2017, pp. 821--832.

\bibitem{pufferfish}
D.~Kifer and A.~Machanavajjhala, ``Pufferfish: A framework for mathematical privacy definitions,'' \emph{ACM Trans. Database Syst.}, vol.~39, no.~1, Jan. 2014.

\bibitem{bayesian_differential_privacy}
B.~Yang, I.~Sato, and H.~Nakagawa, ``Bayesian differential privacy on correlated data,'' in \emph{Proceedings of the 2015 ACM SIGMOD International Conference on Management of Data}, ser. SIGMOD '15.\hskip 1em plus 0.5em minus 0.4em\relax New York, NY, USA: ACM, 2015, p. 747–762.

\bibitem{wang2016mutual}
S.~Wang, L.~Huang, P.~Wang, Y.~Nie, H.~Xu, W.~Yang, X.~Li, and C.~Qiao, ``Mutual information optimally local private discrete distribution estimation,'' \emph{arXiv preprint arXiv:1607.08025}, 2016.

\bibitem{distinguishability_attacks_1}
K.~Chatzikokolakis, G.~Cherubin, C.~Palamidessi, and C.~Troncoso, ``Bayes security: A not so average metric,'' in \emph{2023 IEEE 36th Computer Security Foundations Symposium (CSF)}, 2023, pp. 388--406.

\bibitem{signed_gradient_descent}
J.~Bernstein, Y.-X. Wang, K.~Azizzadenesheli, and A.~Anandkumar, ``signsgd: compressed optimisation for non-convex problems,'' in \emph{International Conference on Machine Learning}, 2018.

\end{thebibliography}


\appendix

\section{Appendix}

\begin{appendices}



    

\section{LDP Mechanisms}\label{section:appendix:ldp_mechanisms}

This section summarises the LDP mechanisms that are used in Section~\ref{section:motivation} and Section~\ref{section:Applications of CPL Analysis}.

\subsubsection{Generalised Randomised Response (GRR)~\cite{LDP_Frequent_Itemset_Mining}}
The GRR mechanism $M$ is a LDP mechanism that maps $X$ stochastically onto itself (i.e., $\mathcal{Y} = \mathcal{X}$), given by,
$p_M(y|x) = \begin{cases} 
\frac{e^\varepsilon}{k - 1 + e^\varepsilon} & \text{if } y = x, \\
\frac{1}{k - 1 + e^\varepsilon} & \text{if } y \neq x.
\end{cases}$
Here, $x \in \mathcal{X}$ and $y \in \mathcal{Y}$, where $\mathcal{X}$ and $\mathcal{Y}$ are alphabets of $X$ and $Y$ respectively. 

\subsubsection{RAPPOR~\cite{rappor_original}}
We adopt basic RAPPOR with unary encoding, where each user's input is first encoded into a binary vector and then perturbed using a two-step randomisation process. The first step---permanent randomised response---flips each bit with parameter $f=0.5$, applied only once. The second step---instantaneous randomisation---perturbs each bit at every data collection using parameters $p=0.5$ and $q=0.75$. It is a common configuration used in pioneer work~\cite{rappor_original}.

\subsubsection{Exponential Mechanism (EXP)~\cite{Algorithmic_Foundations_dwork}}
Given a utility function \( u : \mathcal{X} \times \mathcal{R} \to \mathcal{R} \), the sensitivity of \( u \) is defined as
$\Delta u = \max_{x,x' \in \mathcal{X}} \; \max_{r \in \mathcal{R}} \left| u(x, r) - u(x', r) \right|$.
The exponential mechanism \( M \) selects and outputs an element \( r \in \mathcal{R} \) with probability proportional to
$\exp\left( \frac{\varepsilon \cdot u(x, r)}{2\Delta u} \right)$.
In our experiments, we define the utility function as \( u(x, r) = 1 \) if \( x = r \), and \( u(x, r) = 0 \) otherwise.

\subsubsection{Optimised Unary Encoding (OUE)~\cite{ldp_Frequency_Estimation}}
OUE is a frequency estimation mechanism based on unary encoding, where each input value from a domain of size $d$ is encoded into a binary vector $v$ of length $d$, with a $1$ at the index corresponding to the input and $0$s elsewhere. To preserve privacy, OUE randomises the bits from $v$ independently to generate $y$ as
$    \forall i \in [k] : \quad \Pr[y_i = 1] =
\begin{cases}
\frac{1}{2}; & \text{if } \mathbf{v}_i = 1, \\
\frac{1}{e^\varepsilon+1}; & \text{if } \mathbf{v}_i = 0.
\end{cases}$

\subsubsection{Binary Local Hashing (BLH)~\cite{ldp_Frequency_Estimation} and Optimised Local Hashing (OLH)~\cite{ldp_Frequency_Estimation}}

Let $\mathcal{H}$ be a universal hash function family such that any hash function $H \in \mathcal{H}$ hashes value $v$ into $[g]$ (i.e., $H: v \rightarrow [g]$).  Each user selects a hash function $H \in \mathcal{H}$ and computes $h = H(v)$, which is then privatised using GRR over the range $[g]$. The mechanism outputs $\text{LH}(v) := \langle H, \text{GRR}(h) \rangle$, enabling efficient data collection over large domains with reduced communication. BLH uses $g = 2$, while OLH sets $g = e^\varepsilon + 1$ to minimise estimation variance.

\subsubsection{Summation with Histogram Encoding (SHE)~\cite{ldp_Frequency_Estimation,dwork_she}}
SHE encodes the user’s input \( v \in \mathcal{V} \) as a one-hot vector \( \mathbf{v} \in \{0,1\}^k \) and adds Laplace noise to each component independently: \( \mathbf{y} = \mathbf{v} + \text{Lap}(2/\varepsilon)^k \), satisfying \( \varepsilon \)-LDP.


\subsubsection{Subset Selection (SS)~\cite{wang2016mutual,subsetselection_min}}
SS mechanism obfuscates the true value by reporting a subset \( \Omega \subseteq \mathcal{V} \) of size
$\omega = \max\left(1, \left\lfloor \frac{k}{e^\varepsilon + 1} \right\rfloor\right)$.
The true value \( v \) is included in \( \Omega \) with probability
$p = \frac{\omega e^\varepsilon}{\omega e^\varepsilon + k - \omega}$.
If \( v \in \Omega \), the user samples \( \omega - 1 \) elements uniformly at random from \( \mathcal{V} \setminus \{v\} \); otherwise, \( \omega \) elements are sampled from \( \mathcal{V} \setminus \{v\} \). The user sends \( \Omega \) to the server.

\section{Inferring From Perturbed Output}\label{section:appendix:attacking_models}

This section explains the adopted attribute inference attacks for the LDP mechanisms which are used in statistical CPL estimation.

\subsubsection{GRR and EXP}
Since the output domain of GRR mechanisms is similar to the input domain, we can take the perturbed output as the original value.

\subsubsection{RAPPOR and OUE}

Let $V$ be the set of all possible inputs for the mechanism and $Y$ be perturbed output vector. Then, let $I$ be the set which contains input values corresponding to 1s in $Y$ (i.e., $I = \{v| Y_v=1\},\ v \in V$). Then inferred input from perturbed output $v_{infer}$ given by,
$v_{infer} = Uniform(V)$ if $I = \emptyset$. Otherwise, $v_{infer} = Uniform(I)$~\cite{distinguishability_attacks_1,PETS_statistical_audit}.

\subsubsection{BLH and OLH}

Let $V$ be the set of all possible inputs for the mechanism and $Y$ be perturbed output vector. Then, let $I$ be the set which contains all the input values that hash to $Y$ (i.e., $I = \{v| H(v)=Y\},\ v \in V$). Then inferred input from perturbed output $v_{infer}$ given by,
$v_{infer} = Uniform(V)$ if $I = \emptyset$. Otherwise, $v_{infer} = Uniform(I)$~\cite{distinguishability_attacks_1,PETS_statistical_audit}.

\subsubsection{SS}

Let $V$ be the set of all possible inputs for the mechanism and $\Omega$ be the selected subset sending to the server. Then, let $I$ be the set which contains $v \in V$ values in $\Omega$ (i.e., $I = \{v| v \in \Omega\},\ v \in V$). Then inferred input from perturbed output $v_{infer}$ given by, $v_{infer} = Uniform(I)$~\cite{distinguishability_attacks_1,PETS_statistical_audit}.
\subsubsection{SHE}
In SHE, \( \mathbf{y} \) is released directly. An adversary estimates \( v \) by computing the posterior probability
$P_V(v \mid \mathbf{y}) = \frac{P_Y(\mathbf{y} \mid v) P_V(v)}{\sum_{i=1}^k P_Y(\mathbf{y} \mid i) P_V(i)}$,
where the likelihood is $P_Y(\mathbf{y} \mid v) = \frac{1}{(2b)^k} \exp\left(-\frac{\lVert \mathbf{y} - \mathbf{v} \rVert_1}{b}\right), \ b = 2/\varepsilon$.
Then, Bayes-optimal estimate is $\hat{v} = \arg\max_{v \in \mathcal{V}} P_V(v \mid \mathbf{y})$~\cite{PETS_statistical_audit}.

\section{Monotonicity of Linear Fractional Functions}
Lemma~\ref{lemma:monotone_fn} will be used in mathematical proofs in theoretical analysis.

\begin{lemma}\label{lemma:monotone_fn}
    The function \( f(x) \) defined as $f(x) = \frac{a_1 x + b_1}{a_2 x + b_2}$ is monotonic. Specifically, $f(x)$ is \textbf{increasing} with respect to $x$ if $a_1 b_2 - a_2 b_1 \geq 0$, and $f(x)$ is \textbf{strictly decreasing} with respect to $x$ if $a_1 b_2 - a_2 b_1 < 0$.
\end{lemma}

\begin{proof}
    We have
            $\frac{d f(x)}{d x} = \frac{a_1(a_2x+b_2) - a_2(a_1x + b_1)}{(a_2x+b_2)^2} = \frac{a_1b_2 - a_2b_1}{(a_2x+b_2)^2}$.
    \begin{equation*}
        \text{Therefore,} \quad\quad \frac{d f(x)}{d x} \text{  } 
        \begin{cases} 
        \geq 0 & \text{; if }  a_1b_2 - a_2b_1 \geq 0, \\
        < 0 & \text{; otherwise.}
        \end{cases}
    \end{equation*}
This completes the proof.
\end{proof}

\section{Proof of Proposition~\ref{prop:total_privacy_leakage_from_probability}}\label{proof:prop:total_privacy_leakage_from_probability}

\begin{proof}

Suppose any attribute $\Hat{X} \in X \setminus \{X_k\}$ perturbed by $\Hat{\mathcal{M}}\in \mathcal{M}$ outputs $\Hat{Y}\in Y$. From (\ref{eqn:privacy_leakage_additional_xk}), CPL $L_{\Hat{X} \rightarrow X_k}$ is given by,
\begin{equation}\label{eqn:additional_privacy_leakage_from_x_hat_about_xk}
\begin{split}
    L_{\Hat{X} \rightarrow X_k} 
    &= \ln{\sup_{y\in \mathcal{\Hat{Y}},x,x'\in\mathcal{X}_k} \left(\frac{\sum_{u \in \mathcal{\Hat{X}}} \frac{p(y,u,x)}{p(x)}}{\sum_{v \in \mathcal{\Hat{X}}} \frac{p(y,v,x')}{p(x')}}\right)}.
\end{split}
\end{equation}
We can represent the relationship between $X_k, \Hat{X}$ and $\Hat{Y}$ as a Markov chain ($X_k \rightarrow \Hat{X} \rightarrow \Hat{Y}$) because $\Hat{M}$ generates $\Hat{Y}$ by perturbing $\Hat{X}$. Therefore, $p(y,u,x) = p(y|u)p(u|x)p(x)$, where $y \in \mathcal{Y}, u \in \mathcal{\Hat{X}}, x \in \mathcal{X}_k$. Then, we have,
\begin{equation}\label{eqn:expanded_additional_privacy_leakage_proof}
    \begin{split}
        L_{\Hat{X} \rightarrow X_k} 
        &= \ln{\sup_{y\in \mathcal{\Hat{Y}},x,x'\in\mathcal{X}_k} \left(\frac{\sum_{u \in \mathcal{\Hat{X}}} p( y|u)p(u|x)}{\sum_{v \in \mathcal{\Hat{X}}} p( y|v)p(v|x')}\right)}.
    \end{split}
\end{equation}
    completing the proof.
\end{proof}

\section{Proof of Proposition~\ref{prop:c_max c_min}}\label{proof:prop:c_max c_min}


\begin{proof}
We know that $F(C, G, G') = \frac{\sum_{i=1}^t c_i g_i}{\sum_{i=1}^t c_i g'_i}$, where $c_i \in C, g_i \in G, g'_i \in G',0 < c_i, g_i, g'_i \leq 1 ; \forall i \in [t]$. For any $c\in C, g \in G$ and $g' \in G'$, we can rewrite $F(C, G, G')$ as
    $F(C, G, G') = \frac{cg + A}{cg' + B}$,
where $A = \sum_{a\in C \setminus c,\ v \in G \setminus g} av$ and $B = \sum_{a\in C \setminus c,\ w \in G' \setminus g'} aw$. From Lemma~\ref{lemma:monotone_fn}, $F(\cdot)$ is a monotone function of $c$. The values of $c$ which maximises $F(C,G,G')$ can be determined as, $c=c_{\max}$ if $g B - g' A \geq 0$.  Otherwise $c=c_{min}$.

Therefore, to compute the maximum value of $F(C, G, G')$, it suffices to evaluate the function using a vector $\Tilde{C}$ which consists only the boundary values $c_{min},c_{max}$ (i.e., $\tilde{C} \in \{c_{\min}, c_{\max}\}^t$).

Since $c \leq e^{\varepsilon}\Tilde{c} + \delta ; \forall c,\Tilde{c} \in C$, we have $c_{max} \leq e^{\varepsilon} c_{min} + \delta$.
Therefore, we can write $c_{max}$ as $c_{max}= k_1 c_{min} + k_2$ using $c_{min}$, where $k_1 \in (1, e^{\varepsilon}]$ and $k_2 \in [0, \delta]$.
Then $J(S,c_{min},c_{max},G,G')$ can be written in the form of $I(c_{min},S,k_1,k_2,G,G')$ as
    $I(c_{min},S,k_1, k_2, G,G') = \frac{c_{min}\sum_{i \in [t] \setminus S}g_i+\lambda\sum_{j \in S}g_j}{c_{min}\sum_{i \in [t] \setminus S}g'_i+\lambda\sum_{j \in S}g'_j}$.
Here, $\lambda=(k_1c_{min}+k_2)$.
Let us take $A = \sum_{j \in S}g_j$ and $B = \sum_{j \in S}g'_j$. We know that $\sum_{i \in [t] \setminus S}g_i + \sum_{j \in S}g_j = 1$ because it is a summation of a conditional probability distribution. Then we have $\sum_{i \in [t]\setminus S }g_i = 1-A$. In the same way we can write $\sum_{i \in [t] \setminus S}g'_i = 1-B$. Then we can write $I(\cdot)$ as 
    $N(c_{min},k_1, k_2, A,B) = \frac{c_{min}(1-A)+(k_1c_{min}+k_2)A}{c_{min}(1-B)+(k_1c_{min}+k_2)B}$.
$N(c_{min},k_1, k_2,A,B)$ is monotone over $k_1$ and $k_2$. From Lemma~\ref{lemma:monotone_fn}, we can maximise $N(\cdot)$ by setting $k_1$ as, $k_1=e^{\varepsilon}$ if $A-B \geq 0$. Otherwise $k_1=1$. Similarly, we can maximise $N(\cdot)$ by setting $k_2$ as $k_2=\delta$ if $A-B \geq 0$. Otherwise $k_2=0$.

However, for any $a\in A$ and $b \in B$, $a < b \implies N(k,a,b) \leq 1$ and $a \geq b \implies N(\cdot) \geq 1$. Since we can change $a$ and $b$ by varying $S$, we can always set $a \geq b$ to maximise $N(\cdot)$. Therefore, $A-B \geq 0$ at the maximum $N(\cdot)$, and $k$ should be $k_1=e^{\varepsilon}$ and $k_2=\delta$ to maximise $N(\cdot)$.
Therefore, at the maximum $J(S,c_{min}, c_{max}, G, G')$ the coefficients satisfy $c_{max} = e^{\varepsilon}c_{min} + \delta$.

This completes the proof.
\end{proof}

\section{Proof of Lemma~\ref{lemma:maximum_value_of_linear_fraction}}\label{proof:lemma:maximum_value_of_linear_fraction}


\begin{proof}
    Suppose $a,b \in [t]$ such that $\frac{g_a}{g'_a} \geq \frac{g_b}{g'_b}$. Then, we have $\frac{g_a + g_b}{g'_a + g'_b} - \frac{g_a}{g'_a} = \frac{g_bg'_a - g_ag'_b}{(g'_a+g'_b)g'_a} \leq 0$.
    Therefore, the fraction of sums is less than the maximum individual ratios in the subset. Therefore, $\max_{S} \frac{\sum_{i \in S} g_i}{\sum_{i \in S} g'_i} = \max_{i \in [t]} \frac{g_i}{g'_i}$.

This completes the proof.
\end{proof}

\section{Proof of Corollary~\ref{corollary:max_privacy_leakage_privacy_budget}}\label{proof:corollary:max_privacy_leakage_privacy_budget}


\begin{proof}
    $L_{\Hat{X}\rightarrow X_k}$ becomes maximum when we release the neighbouring attribute $\Hat{X}$ without any perturbation (i.e. $\varepsilon \rightarrow \infty$). Therefore, from (\ref{eqn:H_with_delta}), let us calculate the limit of $H(\cdot)$ when $\varepsilon \rightarrow \infty$ as follows,
        $\lim_{\varepsilon \rightarrow \infty} H(\cdot) = \frac{\sum_{j \in S}g_j}{\sum_{j \in S}g'_j}$.
    Then, $\lim_{\varepsilon \rightarrow \infty}  H^* = \max_{j \in [t]} \frac{g_j}{g'_j}$ from Lemma~\ref{lemma:maximum_value_of_linear_fraction}. Here, $\max_{j \in [t]} \frac{g_j}{g'_j}$ can be written as $\max_{\Hat{x}\in \mathcal{\Hat{X}}} \frac{p(\Hat{x}|x)}{p(\Hat{x}|x')}$ for $x,x'\in \mathcal{X}_k$. To calculate $\lim_{\varepsilon \rightarrow \infty} L_{\Hat{X}\rightarrow X_k}$, we need to calculate the maximum of $H^*$ for all $G,G', \forall x,x' \in \mathcal{X}_k$. 
    Therefore,
    \begin{equation*}
        \lim_{\varepsilon \rightarrow \infty} L_{\Hat{X}\rightarrow X_k} = \max_{\varepsilon} L_{\Hat{X}\rightarrow X_k} = \ln \max_{\Hat{x}\in \mathcal{\Hat{X}},x,x'\in\mathcal{X}_k} \frac{p(\Hat{x}|x)}{p(\Hat{x}|x')},
    \end{equation*}
    where $\mathcal{\Hat{X}}$ and $\mathcal{X}_k$ are alphabets of $\Hat{X}$ and $X_k$ respectively.

    This completes the proof.
\end{proof} 

\section{Analysis of CPL in Extreme Scenarios}\label{section:appendix:Extreme Scenarios}

Analysing how CPL behaves in extreme cases is beneficial to understanding the general scenarios. Therefore, this section analyses the scenarios where CPL could be maximum and minimum for generic ($\varepsilon,\delta$)-LDP mechanisms based on the derived theoretical results in Section~\ref{section:Proposed Algorithm to Compute CPL}.

\subsection{Maximum CPL}\label{section:Maximum CPL}

Maximum CPL depends on two main factors: 
\begin{enumerate}
    \item [$\bullet$] The privacy budget of the LDP mechanism of the neighbouring attribute (i.e., $\Hat{X}$). 
    If $X_k$ and $\Hat{X}$ are two correlated attributes, then we can calculate the maximum CPL using \eqref{eqn:corollary:max_privacy_leakage_privacy_budget} in Corollary~\ref{corollary:max_privacy_leakage_privacy_budget}.
    \item [$\bullet$] The correlation between two attributes (i.e., $X_k$ and $\Hat{X}$). According to the Corollary~\ref{corollary:maximum_additional_privacy_leakage_over_GG'}, $\max_{G,G'} L_{\Hat{X}\rightarrow X_k} = \varepsilon$. 
\end{enumerate}
 
Also, Corollary~\ref{corollary:maximum_additional_privacy_leakage_over_GG'} states that there is \emph{no} need for two attributes to be perfectly correlated to achieve maximum CPL. 
An example scenario is explained in Appendix~\ref{section:appendix:Theoretical Explanation for the Observations of Statistical Analysis}. 
Also, if two variables are perfectly correlated, then CPL is maximum.

\begin{corollary}\label{corollary:maximum_additional_privacy_leakage_over_GG'}
    Let $X_k$ and $\Hat{X}$ be two correlated attributes. Suppose $\Hat{X}$ is perturbed by $\varepsilon$-LDP mechanism $\Hat{\mathcal{M}}$. Let $G = (p(a_1|x)\text{, }\dots\text{, }p(a_t|x)), G' = (p(a_1|x')\text{, }\dots\text{, }p(a_t|x'))$ where $\mathcal{\Hat{X}} = \{a_1,\dots, a_t\}$ and $x,x' \in \mathcal{X}_k$. Here, $\mathcal{\Hat{X}}$ and $\mathcal{X}_k$ are alphabets of  $\Hat{X}$ and $X_k$, respectively. Then, 
    $\max_{G,G'} L_{\Hat{X}\rightarrow X_k} = \varepsilon$.
    This maximum is attained if there exist distributions $G,G'$ such that for all $g \in G$ and $g' \in G'$, in particular, if $g \neq 0$, then $g'= 0$.
\end{corollary}

\begin{proof}
    According to the Algorithm~\ref{algorithm_2_part_2}, if $g \neq 0 \implies g'=0$ for all $g \in G$ and $g' \in G'$ then $A = 1$ and $B=0$. Therefore $\overline{H}^*=\varepsilon$. Thus, we have  $\max_{G,G'} L_{\Hat{X}\rightarrow X_k} = \varepsilon$.
    
    This completes the proof.
\end{proof}




\subsection{Minimum CPL}

Similar to the maximum case, minimum CPL also depends on two main factors: 
\begin{enumerate}
    \item [$\bullet$] the privacy budget of the LDP mechanism of the neighbouring attribute. 

    \item [$\bullet$] The relationship between two attributes. 
\end{enumerate}
As stated in Corollary~\ref{corollary:min_privacy_leakage_privacy_budget}, there is zero CPL when either the privacy budget of the neighbouring attribute (i.e., $\Hat{X}$) is $\varepsilon = 0$ or two attributes are independent.

\begin{corollary}\label{corollary:min_privacy_leakage_privacy_budget}
    Let $X_k$ and $\Hat{X}$ be two correlated attributes. Suppose $\Hat{X}$ is perturbed by $\varepsilon$-LDP mechanism $\Hat{\mathcal{M}}$. Then, there is no CPL caused by $\Hat{X}$ about $X_k$ when either $\varepsilon =0$ or two attributes are independent.
\end{corollary}

\begin{proof}
    First, let us consider the $\varepsilon=0$ scenario. From Algorithm~\ref{algorithm_2_part_2}, we have $\overline{H}^*=1$ as $e^\varepsilon = 1$. Therefore, we have zero CPL (i.e., $L_{\Hat{X}\rightarrow X_k} = 0$) from Algorithm~\ref{algorithm_2_part_2}. This means there is no CPL between any attributes when $\varepsilon = 0$.
    
    Next, let us consider the scenario where two attributes (i.e., $X_k$ and $\Hat{X}$) are independent. Then, we know that $G=G'$, where $G = (p(a_1|x)\text{, } p(a_2|x)\text{,}\dots\text{, }p(a_t|x))$, $G' = (p(a_1|x')\text{, }p(a_2|x')\text{,}\dots\text{, } p(a_t|x'))$, $\mathcal{\Hat{X}} = \{a_1, a_2, \dots, a_t\}$ and $x,x' \in \mathcal{X}_k$. Therefore, from Algorithm~\ref{algorithm_2_part_2}, we have $\overline{H}^*=1$ as $e^\varepsilon = 1$ and we have $L_{\Hat{X}\rightarrow X_k} = 0$ from Algorithm~\ref{algorithm_2_part_2}. Thus, there is no CPL when attributes are independent.
    
    This completes the proof.
\end{proof}

\section{Proof of Lemma~\ref{lemma:sigma_greater_than_ratio}}\label{proof:lemma:sigma_greater_than_ratio}


\begin{proof}
Let us take $\Delta$ as
    $\Delta = \frac{1 + (A + \sigma g')(e^\varepsilon-1)}{1 + (B + g')(e^\varepsilon-1)} - \frac{1 + A(e^\varepsilon-1)}{1 + B(e^\varepsilon-1)}$.
By simplifying, we have,
    $\Delta = \frac{g'(e^\varepsilon-1) \left(\sigma + \sigma B(e^\varepsilon-1) - 1 - A(e^\varepsilon-1)\right)}{\left(1 + (B + g')(e^\varepsilon-1)\right)\left(1 + B(e^\varepsilon-1)\right)}.$
Since the denominator and $g'(e^\varepsilon-1)$ are always positive, the sign of $\Delta$ only depends on $\left(\sigma + \sigma B(e^\varepsilon-1) - 1 - A(e^\varepsilon-1)\right)$.
Therefore,
\begin{equation}\label{eqn:Delta_sign}
    \Delta \text{  } 
    \begin{cases} 
    \geq 0 & \text{; if }  \sigma + \sigma B(e^\varepsilon-1) - 1 - A(e^\varepsilon-1) \geq 0, \\
    < 0 & \text{; otherwise.}
    \end{cases}
\end{equation}
If $\sigma + \sigma B(e^\varepsilon-1) - 1 - A(e^\varepsilon-1) \geq 0$, we have:
    $\sigma \geq \frac{1 + A(e^\varepsilon-1)}{1 + B(e^\varepsilon-1)}$.

This completes the proof.
\end{proof}

\section{Proof of Theorem~\ref{Theorem:optimality_of_algo}}\label{proof:Theorem:optimality_of_algo}


\begin{proof}
    Let $Q = \{q_1,q_2,\dots,q_t\}$ be the elements of element-wise division $G\oslash G'$ sorted in non-increasing order, i.e. $q_1 \geq q_2 \geq \dots \geq q_t$. Here, $t$ is the vector size of $G$ and $G'$.
    Next, let us recall the Algorithm~\ref{algorithm_2_part_2}.
    There are two variables, $A$ and $B$. The algorithm traverse all $q \in Q$ and adds $g \in G$ and $g' \in G'$ corresponding to $q$ to $A,B$ respectively if $q \geq \frac{1+A(exp(\varepsilon)-1)}{1+B(exp(\varepsilon)-1)}$. Suppose $S = \{q_1,q_2,\dots,q_v\}$ is the subset of $Q$ selected to update $A$ and $B$ by the algorithm. Here $v \leq t$. Let $H$ be the computed objective value of $H(\cdot)$ by the algorithm for $S$.

    First, we prove by contradiction that no other subset of $S$ can give a strictly larger objective value than $H$.\\

    \textbf{Assume there exists $\Hat{S} \subset S$ such that gives objective value $\Hat{H} (> H)$: } 
    Since $\Hat{S} \subset S$, choose any $\Hat{q} \in S \setminus \Hat{S}$. Since $Q$ is in non-increasing order, $q_1 \geq \Hat{q} \geq q_v$ where $q_v$ is the smallest element of $S$. We know that $q_v$ is included in calculating $H$ as $q_v \in S$. Therefore, from Lemma~\ref{lemma:less_than_a_b}, we have $H \leq q_v$. Suppose $H'$ is the value of the objective function computed using $S \setminus \{q'\}$. Here, $q' \in Q (\geq q_v)$. Then, from Lemma~\ref{lemma:sigma_greater_than_ratio}, we know $H' \leq H$. Since $\Hat{H}$ is the value of the objective function computes with $\Hat{S} \subset S$, we have $\Hat{H} \leq H' \leq H$. This is a contradiction with the initial assumption.  Consequently, no smaller subset $\Hat{S} \subset S$ can give a strictly greater value than $H$. 

    \textbf{Assume there exists $\Tilde{S} \subseteq Q$ such that gives objective value  $\Tilde{H} (> H)$: } 
    Let us take $\Hat{S} = S \cap \Tilde{S}$ and $\Hat{H}$ as the value of the objective function computed on $\Hat{S}$. Since $\Hat{S} \subseteq S$, we know that $\Hat{H} \leq H$ from the previous paragraph. Next, let us consider any $\Tilde{q} \in \Tilde{S} \setminus S$. Then, from Lemma~\ref{lemma:sigma_greater_than_ratio}, we know that there should be at least one such $\Tilde{q} > H$ to make $\Tilde{H} > H$ because $\Hat{H} \leq H$. However, from the terminating condition (line 8) of the Algorithm~\ref{algorithm_2_part_2} and Lemma~\ref{lemma:sigma_greater_than_ratio}, we know that $H > \Tilde{q}$ for all $\Tilde{q}$. Therefore, $\Hat{H} < H$. This is a contradiction with the initial assumption. Consequently, no any subset $\Tilde{S} \subseteq Q$ can give a strictly greater value than $H$.

    This completes the proof.
\end{proof}

\begin{table*}[ht]
\centering
\begin{minipage}
{0.3\textwidth}
\centering
  \begin{tabular}{|l|c|c|c|}
\hline
Dataset & $\varepsilon$ & \multicolumn{2}{|c|}{NMSE-CPL} \\
\cline{3-4} 
  &  & GRR ($\times 10^{-4}$) & EXP ($\times 10^{-4}$) \\\hline
 
 & 1 & 0.889 & 3.36 \\ 
\cline{2-4} 
 CelebA & 3 & 0.197 & 0.514  \\
 \cline{2-4} 
\hline
 & 1 & 2.99 & 8.16 \\ 
\cline{2-4} 
Adult & 3 & 3.08 & 2.73 \\ 
\cline{2-4} 
\hline
 & 1 & 11.6 & 54.1 \\ 
\cline{2-4} 
Cardio & 3 & 2.49 & 5.82 \\
\cline{2-4} 
\hline
 & 1 & 0.674 & 3.79  \\ 
\cline{2-4} 
DSS & 3 & 0.168 & 0.241   \\
 \cline{2-4} 
\hline
\end{tabular}
\caption{NMSE-CPL error between statistical and theoretical (Algorithm~\ref{algorithm_1}) CPL values.}
\label{table:algo_1_vs_statistical}
\end{minipage}
\hfill
\begin{minipage}{0.6\textwidth}
\centering
  \begin{tabular}{|l|c|c|c|c|c|c|c|c|c|}
\hline
Dataset & $\varepsilon$ & \multicolumn{8}{|c|}{NMSE-CPL} \\
\cline{3-10} 
  &  & GRR & RAPPOR  & EXP & OUE & BLH & OLH & SS & SHE  \\\hline
 
 & 1 & 8.89e-5 & 0.29 & 0.29 & 0.32 & 0.32 & 0.33 & 9.44e-5 & 0.29 \\ 
\cline{2-10} 
 CelebA & 3 & 1.97e-5 & 0.31 & 0.32 & 0.54 & 0.54 & 0.53 & 1.81e-5 & 0.35  \\
 \cline{2-10} 
\hline
 & 1 & 0.017 & 0.25 & 0.30 & 0.22 & 0.34 & 0.22 & 0.082 & 0.25 \\ 
\cline{2-10} 
Adult & 3 & 0.0037 & 0.23 & 0.29 & 0.31 & 0.55 & 0.31 & 0.0037 & 0.27 \\ 
\cline{2-10} 
\hline
 & 1 & 0.0088& 0.25 & 0.29 & 0.25 & 0.31 & 0.26 & 0.038 & 0.26\\ 
\cline{2-10} 
Cardio & 3 & 0.00075 & 0.23 & 0.26 & 0.40 & 0.50 & 0.39 & 0.00060 & 0.26\\
\cline{2-10} 
\hline
 & 1 & 0.031 & 0.28 & 0.35 & 0.19 & 0.41 & 0.19 & 0.17 & 0.29 \\ 
\cline{2-10} 
DSS & 3 & 0.012 & 0.25 & 0.33 & 0.23 & 0.63 & 0.22 & 0.012 & 0.28  \\
 \cline{2-10} 
\hline
\end{tabular}
\caption{NMSE-CPL errro between statistical and theoretical (Algorithm~\ref{algorithm_2_part_2}) CPL values.}\label{table:algo_2_vs_statistical}
\end{minipage}
\end{table*}

\begin{lemma}\label{lemma:less_than_a_b}
Let $0 \leq A, B, g,g' \leq 1$, and suppose $\frac{g}{g'} \geq \frac{1+A(e^\varepsilon-1)}{1+B(e^\varepsilon-1)}$ and $e^\varepsilon > 1$. Then,
       $\frac{1+(A+g)(e^\varepsilon-1)}{1+(B+g')(e^\varepsilon-1)} - \frac{g}{g'} \leq 0$.
\end{lemma}

\begin{proof}
    Let us take $f(x) = \frac{1+(A+g)(e^\varepsilon-1)}{1+(B+g')(e^\varepsilon-1)} - \frac{g}{g'}$. Since $\frac{1+A(e^\varepsilon-1)}{1 + B(e^\varepsilon-1)} \leq \frac{g}{g'}$, $f(x) \leq 0$.
    $\frac{1+(A+g)(e^\varepsilon-1)}{1+(B+g')(e^\varepsilon-1)} - \frac{g}{g'} \leq 0$
    This completes the proof.
\end{proof}

\section{Maximum CPL under a non-Perfect Correlation}\label{section:appendix:Theoretical Explanation for the Observations of Statistical Analysis}

The following example illustrates how CPL could become maximum even under a non-perfect correlation which highlights the privacy risk of the conventional correlation metrics usage
to quantify CPL. Suppose $\mathcal{X}_k = \{x_1, x_2, x_3, x_4\},\mathcal{\Hat{X}} = \{\Hat{x}_1, \Hat{x}_2, \Hat{x}_3, \Hat{x}_4\}$ are alphabets of $X_k$ and $\Hat{X}$, respectively and (\ref{eqn:joint_prob_distribution}) is the joint probability distribution of two attributes.
If we take $G = P(\Hat{X}| X_k = x_1)$ and $G' = P(\Hat{X}| X_k = x_2)$, then this satisfy the maximum privacy leakage requirement in Corollary~\ref{corollary:maximum_additional_privacy_leakage_over_GG'}. Therefore, $L_{\Hat{X} \rightarrow X_k} = \varepsilon$. However, note that $X_k$ and $\Hat{X}$ are not perfectly correlated. According to conventional metrics, $X_k$ and $\Hat{X}$  are \textbf{weakly} correlated as {normalised MI  = 0.164} and {PCC = 0.357} (i.e., $0.2 \leq PCC < 0.4$ weak correlation~\cite{Wikipedia_Correlation_Coefficient_thresholds}). We compare these theoretical results with statistically estimated CPL as summarised in Table~\ref{table:maximum_leakage_example}. Here, we generate a synthetic dataset of 100K samples using \eqref{eqn:joint_prob_distribution}.

\begin{equation}\label{eqn:joint_prob_distribution}
        P(X_k, \Hat{X}) = 
     \bordermatrix{ 
        &\Hat{x}_1 &\Hat{x}_2 &\Hat{x}_3 &\Hat{x}_4 \cr
       x_1 &0.2 &0 &0 &0 \cr
       x_2 &0 &0.2 &0 &0 \cr
       x_3 &0.1 &0.15 &0.03 &0.02 \cr
       x_4 &0.1 &0.15 &0.03 &0.02 \cr}
\end{equation}

\begin{table}[h]
\centering
\begin{tabular}{|c|c|c|c|c|}
\hline
Privacy Budget & \multicolumn{2}{|c|}{Experimental (GRR)} & \multicolumn{2}{|c|}{Theoretical} \\
\cline{2-5}
 $\varepsilon$ & $L_{\Hat{X} \rightarrow X_k}$ & $L_{X_k \rightarrow \Hat{X}}$ & $L_{\Hat{X} \rightarrow X_k}$ & $L_{X_k \rightarrow \Hat{X}}$ \\\hline
$0.5$ & 0.5007  & 0.2833 & 0.5 & 0.2810 \\
$1$ & 0.9985 & 0.6222 & 1 & 0.6203  \\
$2$ & 2.0020 & 1.4369 & 2 & 1.4340  \\
\hline
\end{tabular}
\caption{Experimental (GRR) and theoretical values of $\max_{G,G'} L_{\Hat{X}\rightarrow X_k}$ and $\max_{G,G'} L_{X_k \rightarrow \Hat{X}}$.}
\label{table:maximum_leakage_example}
\end{table}




\section{Results: Statistical vs Theoretical (Algorithm~\ref{algorithm_1})}\label{section:Results: Statistical vs Theoretical algo 1}

Table~\ref{table:algo_1_vs_statistical} summarises NMSE-CPL between the CPL given by the statistical method and Algorithm~\ref{algorithm_1} for GRR and EXP mechanisms.  Since all the NMSE-CPL errors are very small ($<5.5\times 10^{-3}$), which validates the computed CPL values of Algorithm~\ref{algorithm_1}.

\textbf{NMSE-CPL} quantifies the normalised mean square error between theoretical CPL $L^*_{x \rightarrow x_k}$ and statistically estimated CPL $\tilde{L}_{x \rightarrow x_k}$ for all the attributes in the dataset as, 
\begin{equation}
    \operatorname{NMSE-CPL} = \frac{\sum_{x_k \in X} \sum_{x \in X\setminus \{x_k\}} \bigl(\tilde{L}_{x \rightarrow x_k} - L^*_{x \rightarrow x_k}\bigr)^2}{\sum_{x_k \in X} \sum_{x \in X\setminus \{x_k\}}( L^*_{x \rightarrow x_k})^2}.
\end{equation}
Here, $n$ denotes the total number of attributes in set $X$.

\section{Results: Statistical vs Theoretical (Algorithm~\ref{algorithm_2_part_2})}\label{section:Results: Statistical vs Theoretical algo 2}

In this section, we calculate NMSE-CPL between the CPL given by the statistical method and the theoretical method (Algorithm~\ref{algorithm_2_part_2}). Table~\ref{table:algo_2_vs_statistical} summarises the results. Results reveal that both GRR and SS mechanisms leak CPL close to the optimal upper bound, which can be computed using Algorithm~\ref{algorithm_2_part_2}.

\section{Privacy Budget Calibration}\label{section:appendix:privacy_budget_calibration}


When all attributes are assigned the same privacy budget, the calibration problem can be formulated as the following optimisation:

\begin{equation}
    \begin{split}
        \textbf{maximise} & \quad \tilde{\varepsilon}\\
        \textbf{subject to} & \quad \tilde{\varepsilon} + \sum_{x \in X\setminus \{X_k\}} L_{x\rightarrow X_k} \leq \overline{\varepsilon}; \quad\forall k \in [n].
    \end{split}
\end{equation}

where $\tilde{\varepsilon}$ is the per-attribute calibrated budget, $\overline{\varepsilon}$ is the maximum allowable privacy leakage, and $L_{x \rightarrow X_k}$ denotes the privacy leakage from attribute $x$ to attribute $X_k$. We solve this optimisation problem using signed gradient descent~\cite{signed_gradient_descent} with a step size $\Delta = 0.01$, which efficiently converges to the optimal solution, because we know that $L_{x \rightarrow X_k}$ monotonically increases with $\tilde{\varepsilon}$.

\begin{algorithm}[h]
\caption{Calibrating Privacy Budgets.}\label{algorithm: budget_calibration}
\begin{algorithmic}[1]
    \Statex \textbf{Input:}
    \Statex \hspace{1em} $\overline{\varepsilon}$ \Comment{\textit{Total privacy budget.}}
    \Statex \hspace{1em} $n$ \Comment{\textit{Attribute count.}}
    \Statex \hspace{1em} $P(X_i,X_j), \ \forall i,j \in [n]$
    \Comment{\textit{Probability distributions.}}
    \Statex \hspace{1em} $\Delta$\Comment{\textit{Step size.}}
    \State $\tilde{\varepsilon} \gets \overline{\varepsilon}/n$ \Comment{\textit{Initial $\overline{\varepsilon}$ based on SPL algorithm.}}
    \State $L \gets 0$ \Comment{\textit{To keep the maximum TPL.}}
    \While{$L < \overline{\varepsilon}$} 
    \Comment{\textit{To keep current TPL}}
    \ForAll{$i \in [n]$}
    \State$ l \gets 0$ 
    \ForAll{$j \in [n]$}
    \State$L_{j \rightarrow i} \gets \text{Compute CPL($\tilde{\varepsilon},P(X_i,X_j)$)}$ 
        \State$ l \gets l + L_{j \rightarrow i}$ \Comment{\textit{Update $l$.}}
    \EndFor
    \State $L \gets \max (l,L)$
    \EndFor
    \State $ \tilde{\varepsilon} \gets \tilde{\varepsilon} + \Delta$ \Comment{\textit{Update $\tilde{\varepsilon}$}.}
    \EndWhile
    \State \Return $\varepsilon^* = \tilde{\varepsilon}$
\end{algorithmic}
\end{algorithm}



\end{appendices}

\end{document}